\documentclass[lettersize,journal]{IEEEtran}
\usepackage{array}
\usepackage{soul}
\usepackage{textcomp}
\usepackage{stfloats}
\usepackage{url}
\usepackage{verbatim}
\usepackage{graphicx}
\usepackage{cite}
\usepackage{numprint}
\npthousandsep{\,}
\usepackage[inline]{enumitem}
\usepackage{tcolorbox}
\usepackage{float}
\usepackage[font=small,skip=0pt]{caption}
\usepackage{makecell}
\usepackage{arydshln}
\usepackage{times}
\usepackage{amsmath,bm,amssymb,amsthm,amsfonts}
\usepackage{listings}%
\usepackage[ruled,linesnumbered]{algorithm2e}
\usepackage[italic]{mathastext} 
\usepackage{graphicx}%
\usepackage{booktabs}
\usepackage{setspace,hyperref}
\usepackage{color}
\usepackage{etoolbox}
\usepackage{multirow}
\usepackage{circuitikz}
\usepackage{tikz}
\usepackage{pgfkeys}
\usepackage{cuted}
\usepackage{mathtools}
\usepackage[labelformat=simple]{subcaption}

\usepackage{multirow}
\usepackage[normalem]{ulem}
\useunder{\uline}{\ul}{}
\usepackage{cancel}
\usepackage{tkz-tab}
\usepackage{latexsym}
\usepackage{subcaption}

\usetikzlibrary{shapes.geometric, arrows}
\usetikzlibrary{positioning}
\usetikzlibrary{shapes.symbols}
\usetikzlibrary{angles,quotes}

\newtheorem{thm}{Theorem}
\newtheorem{proposition}{Proposition}

\newtheorem{lemma}{Lemma}

\definecolor{darkgreen}{rgb}{0.1333,0.545,0.1333}
\definecolor{darkred}{rgb}{0.698,0,0}
\hypersetup{colorlinks=true,%
            breaklinks=true,%
            linkcolor=black,%
            pdfstartview=FitH,
            pdfpagelayout=TwoColumnRight}%
\newcommand{\RNum}[1]{\uppercase\expandafter{\romannumeral #1\relax}}
\tikzstyle{box} = [rectangle, rounded corners, minimum width=3cm, minimum height=1cm,text centered, draw=black, fill=red!30,text width=4.5cm]
\tikzstyle{boxwithin} = [rectangle, rounded corners, minimum width=3cm, minimum height=1cm,text centered, draw=black, fill=cyan,text width=4.5cm]
\tikzstyle{decision} = [diamond,aspect=5, minimum width=2cm, minimum height=0.5cm, text centered, draw=black, fill=green!30,text width=3.5cm]
\tikzstyle{arrow} = [thick,->,>=stealth]
\tikzset{
    block/.style={rectangle, draw, line width=0.5mm, black, text width=5em, text centered,                 rounded corners, minimum height=2em},
    line/.style={draw, -latex}
}
\input{irs.tex}
\input{basestation.tex}

\makeatletter

\newif\ifleftantennas

\pgfkeys{
  /mimonode/leftantennas/.is if=leftantennas,
  /pgf/.cd,
  right antennas/.code={\pgfkeys{/mimonode/leftantennas=false}},
  left antennas/.code={\pgfkeys{/mimonode/leftantennas=true}},
  circle fill color/.initial=white,
  circle stroke color/.initial=black,
  circle radius/.initial=2mm,
  antenna offset/.initial=0.3cm,
  antenna yshift/.initial=0.cm,
  antenna base height/.initial=0.2 cm,
  antenna side/.initial=0.4cm,
}

\def\mimoshapesdrawantenna{
  \pgf@xa=\pgf@x 
  \pgf@ya=\pgf@y
  \pgfpathmoveto{\pgfpoint{\pgf@x}{\pgf@y}}

  \advance\pgf@xa by \antennaoffset
  \pgfpathlineto{\pgfpoint{\pgf@xa}{\pgf@ya}}
  \pgfpathclose
  \pgfmoveto{\pgfpoint{\pgf@xa}{\pgf@ya}}

  \pgf@xb=\pgf@xa \pgf@yb=\pgf@ya
  \advance\pgf@yb by \antennabaseheight cm
  \pgfpathlineto{\pgfpoint{\pgf@xb}{\pgf@yb}}
  
  \pgf@process{
    \pgfpointadd{\pgfpoint{\pgf@xb}{\pgf@yb}}{\pgfpointpolar{60}{\pgfkeysvalueof{/pgf/antenna side}}}
  }
  \pgf@xc=\pgf@x 
  \pgf@yc=\pgf@y
  \pgfpathlineto{\pgfpoint{\pgf@x}{\pgf@y}}

  \pgf@process{
    \pgfpointadd{\pgfpoint{\pgf@xc}{\pgf@yc}}{\pgfpointpolar{180}{\pgfkeysvalueof{/pgf/antenna side}}}
  }
  \pgfpathlineto{\pgfpoint{\pgf@x}{\pgf@y}}
  
  \pgf@process{
    \pgfpointadd{\pgfpoint{\pgf@xb}{\pgf@yb}}{\pgfpointpolar{240}{\pgfkeysvalueof{/pgf/antenna side}}}
  }
  \pgfpathlineto{\pgfpoint{\pgf@xb}{\pgf@yb}}
}

\def\drawellipses{
  \pgf@xa=\pgf@x 
  \advance\pgf@xa by \antennaoffset
  \pgf@ya=\pgf@y
  
  \pgf@process{\pgfpathcircle{\pgfpoint{\pgf@xa}{\pgf@ya+\antennabaseheight-1pt}}{1pt}}
  \pgf@process{\pgfpathcircle{\pgfpoint{\pgf@xa}{\pgf@ya+\antennabaseheight+2pt}}{1pt}}
  \pgf@process{\pgfpathcircle{\pgfpoint{\pgf@xa}{\pgf@ya+\antennabaseheight+5pt}}{1pt}}
}

\def\antennaStartXCoordinate{
  \ifleftantennas
  \pgf@x=-.5\wd\pgfnodeparttextbox
  \pgfmathsetlength\pgf@xc{-\pgfkeysvalueof{/pgf/inner xsep}}
  \advance\pgf@x by \pgf@xc%
  \setlength{\pgf@xa}{-\pgfshapeminwidth}
  \ifdim\pgf@x>.5\pgf@xa
  \pgf@x=.5\pgf@xa
  \fi
  \else
  \pgf@x=.5\wd\pgfnodeparttextbox 
  \pgfmathsetlength\pgf@xc{\pgfkeysvalueof{/pgf/inner xsep}}
  \advance\pgf@x by \pgf@xc%
  \setlength{\pgf@xa}{\pgfshapeminwidth}
  \ifdim\pgf@x<.5\pgf@xa
  \pgf@x=.5\pgf@xa
  \fi
  \fi
}

\pgfdeclareshape{mimoone}{

  \savedmacro\antennaoffset{
    \ifleftantennas
      \def\antennaoffset{-\pgfkeysvalueof{/pgf/antenna offset}}
    \else
      \def\antennaoffset{\pgfkeysvalueof{/pgf/antenna offset}}
    \fi
  }

  \savedmacro\antennabaseheight{
    \def\antennabaseheight{\pgfkeysvalueof{/pgf/antenna base height}}
  }

  \savedmacro\antennayshift{
    \def\antennayshift{\pgfkeysvalueof{/pgf/antenna yshift}}
  }

  \savedmacro\shapewidth{
    \def\shapewidth{\pgfkeysvalueof{/pgf/minimum width}}
  }

  \savedmacro\shapeheight{
    \def\shapeheight{\pgfkeysvalueof{/pgf/minimum height}}
  }

  \savedanchor\northeast{%
    \pgf@y=.5\ht\pgfnodeparttextbox 
    \pgf@x=.5\wd\pgfnodeparttextbox 
    \pgfmathsetlength\pgf@xc{\pgfkeysvalueof{/pgf/inner xsep}}
    \advance\pgf@x by \pgf@xc%
    \pgfmathsetlength\pgf@yc{\pgfkeysvalueof{/pgf/inner ysep}}
    \advance\pgf@y by \pgf@yc%
    \setlength{\pgf@xa}{\shapewidth}
    \ifdim\pgf@x<.5\pgf@xa
    \pgf@x=.5\pgf@xa
    \fi
    \setlength{\pgf@ya}{\shapeheight}
    \ifdim\pgf@y<.5\pgf@ya
    \pgf@y=.5\pgf@ya
    \fi
    \pgfmathsetlength\pgf@xa{\pgfkeysvalueof{/pgf/outer xsep}}%
    \advance\pgf@x by\pgf@xa%
    \pgfmathsetlength\pgf@ya{\pgfkeysvalueof{/pgf/outer ysep}}%
    \advance\pgf@y by\pgf@ya%
  }

  \savedanchor\southwest{%
    \pgf@y=-.5\ht\pgfnodeparttextbox 
    \pgf@x=-.5\wd\pgfnodeparttextbox 
    \pgfmathsetlength\pgf@xc{-\pgfkeysvalueof{/pgf/inner xsep}}
    \advance\pgf@x by \pgf@xc%
    \pgfmathsetlength\pgf@yc{-\pgfkeysvalueof{/pgf/inner ysep}}
    \advance\pgf@y by \pgf@yc%
    \setlength{\pgf@xa}{-\shapewidth}
    \ifdim\pgf@x>.5\pgf@xa
    \pgf@x=.5\pgf@xa
    \fi
    \setlength{\pgf@ya}{-\shapeheight}
    \ifdim\pgf@y>.5\pgf@ya
    \pgf@y=.5\pgf@ya
    \fi
    \pgfmathsetlength\pgf@xa{-\pgfkeysvalueof{/pgf/outer xsep}}%
    \advance\pgf@x by\pgf@xa%
    \pgfmathsetlength\pgf@ya{-\pgfkeysvalueof{/pgf/outer ysep}}%
    \advance\pgf@y by\pgf@ya%
  }

  \savedanchor{\anchorA}{
    \pgf@y=.5\ht\pgfnodeparttextbox
    \setlength{\pgf@ya}{\pgfshapeminheight}
    \ifdim\pgf@y<.5\pgf@ya
    \pgf@y=.5\pgf@ya
    \fi
    \pgf@y=0.\pgf@y
    \advance\pgf@y by \antennayshift
    \antennaStartXCoordinate
  }

  \inheritanchorborder[from=rectangle]
  
  \inheritanchor[from=rectangle]{center}
  \inheritanchor[from=rectangle]{north}
  \inheritanchor[from=rectangle]{east}
  \inheritanchor[from=rectangle]{south}
  \inheritanchor[from=rectangle]{west}
  \inheritanchor[from=rectangle]{north east}
  \inheritanchor[from=rectangle]{north west}
  \inheritanchor[from=rectangle]{south west}
  \inheritanchor[from=rectangle]{south east}

  \anchor{A}{
    \pgf@process{\anchorA}
  }

  \anchor{text}
  {
    \pgf@process{\pgfpointorigin}
    \advance\pgf@x by -.5\wd\pgfnodeparttextbox%
    \advance\pgf@y by -.5\ht\pgfnodeparttextbox%
    \advance\pgf@y by +.5\dp\pgfnodeparttextbox%
  }
  
  \anchor{first antenna base start}{
    \pgf@process{\anchorA}
  }

  \anchor{first antenna base end}{
    \pgf@process{\anchorA}
    \advance\pgf@x by \antennaoffset%
    \advance\pgf@y by 0.2cm%
  }

  \anchor{first antenna}{
    \pgf@process{\anchorA}
    \advance\pgf@x by \antennaoffset%
    \advance\pgf@y by 0.4cm%
  }

  \backgroundpath{
    \pgfpathrectanglecorners
    { 
      \pgfpointadd{\southwest}{\pgfpoint{\pgfkeysvalueof{/pgf/outer xsep}}{\pgfkeysvalueof{/pgf/outer ysep}}}
    }
    { 
      \pgfpointadd{\northeast}{\pgfpointscale{-1}{\pgfpoint{\pgfkeysvalueof{/pgf/outer xsep}}{\pgfkeysvalueof{/pgf/outer ysep}}}}
    }

    \pgf@process{\anchorA}
    \mimoshapesdrawantenna
  }
}

\pgfdeclareshape{mimotwo}{

  \inheritsavedanchors[from=mimoone]  

  \savedanchor{\anchorA}{
    \pgf@y=.5\ht\pgfnodeparttextbox
    \setlength{\pgf@ya}{\pgfshapeminheight}
    \ifdim\pgf@y<.5\pgf@ya
    \pgf@y=.5\pgf@ya
    \fi
    \pgf@y=0.4\pgf@y
    \advance\pgf@y by \antennayshift
    \antennaStartXCoordinate
  }

  \savedanchor{\anchorB}{
    \pgf@y=-.5\ht\pgfnodeparttextbox
    \setlength{\pgf@ya}{-\pgfshapeminheight}
    \ifdim\pgf@y>.5\pgf@ya
    \pgf@y=.5\pgf@ya
    \fi
    \pgf@y=0.6\pgf@y
    \advance\pgf@y by \antennayshift
    \antennaStartXCoordinate
  }

  \inheritanchorborder[from=mimoone]
  
  \inheritanchor[from=mimoone]{center}
  \inheritanchor[from=mimoone]{north}
  \inheritanchor[from=mimoone]{east}
  \inheritanchor[from=mimoone]{south}
  \inheritanchor[from=mimoone]{west}
  \inheritanchor[from=mimoone]{north east}
  \inheritanchor[from=mimoone]{north west}
  \inheritanchor[from=mimoone]{south west}
  \inheritanchor[from=mimoone]{south east}
  \inheritanchor[from=mimoone]{text}
  \inheritanchor[from=mimoone]{A}
  \inheritanchor[from=mimoone]{first antenna base start}
  \inheritanchor[from=mimoone]{first antenna base end}
  \inheritanchor[from=mimoone]{first antenna}
  
  \anchor{B}{
    \pgf@process{\anchorB}
  }

  \anchor{second antenna base start}{
    \pgf@process{\anchorB}
  }

  \anchor{second antenna base end}{
    \pgf@process{\anchorB}
    \advance\pgf@x by \antennaoffset%
    \advance\pgf@y by 0.2cm%
  }

  \anchor{second antenna}{
    \pgf@process{\anchorB}
    \advance\pgf@x by \antennaoffset%
    \advance\pgf@y by 0.4cm%
  }

  \backgroundpath{
    \pgfpathrectanglecorners
    { 
      \pgfpointadd{\southwest}{\pgfpoint{\pgfkeysvalueof{/pgf/outer xsep}}{\pgfkeysvalueof{/pgf/outer ysep}}}
    }
    { 
      \pgfpointadd{\northeast}{\pgfpointscale{-1}{\pgfpoint{\pgfkeysvalueof{/pgf/outer xsep}}{\pgfkeysvalueof{/pgf/outer ysep}}}}
    }

    \pgf@process{\anchorA}
    \mimoshapesdrawantenna

    \pgf@process{\anchorB}
    \mimoshapesdrawantenna
  }
}

\pgfdeclareshape{mimothree}{

  \inheritsavedanchors[from=mimotwo]

  \savedanchor{\anchorA}{
    \pgf@y=.5\ht\pgfnodeparttextbox
    \setlength{\pgf@ya}{\pgfshapeminheight}
    \ifdim\pgf@y<.5\pgf@ya
    \pgf@y=.5\pgf@ya
    \fi
    \pgf@y=0.5\pgf@y
    \advance\pgf@y by \antennayshift
    \antennaStartXCoordinate
  }

  \savedanchor{\anchorC}{
    \pgf@y=-.5\ht\pgfnodeparttextbox
    \setlength{\pgf@ya}{-\pgfshapeminheight}
    \ifdim\pgf@y>.5\pgf@ya
    \pgf@y=.5\pgf@ya
    \fi
    \pgf@y=0.9\pgf@y
    \advance\pgf@y by \antennayshift
    \antennaStartXCoordinate
  }

  \savedanchor{\anchorB}{
    \pgf@y=.5\ht\pgfnodeparttextbox
    \setlength{\pgf@ya}{\pgfshapeminheight}
    \ifdim\pgf@y<.5\pgf@ya
    \pgf@y=.5\pgf@ya
    \fi
    \pgf@ya=0.5\pgf@y
    \pgf@y=-.5\ht\pgfnodeparttextbox
    \setlength{\pgf@yc}{-\pgfshapeminheight}
    \ifdim\pgf@y>.5\pgf@yc
    \pgf@y=.5\pgf@yc
    \fi
    \pgf@yc=0.9\pgf@y
    \pgf@y=0.5\pgf@yc
    \advance\pgf@y by 0.5\pgf@ya
    \advance\pgf@y by \antennayshift
    \antennaStartXCoordinate
  }

  \inheritanchorborder[from=mimotwo]

  \inheritanchor[from=mimotwo]{center}
  \inheritanchor[from=mimotwo]{north}
  \inheritanchor[from=mimotwo]{east}
  \inheritanchor[from=mimotwo]{south}
  \inheritanchor[from=mimotwo]{west}
  \inheritanchor[from=mimotwo]{north east}
  \inheritanchor[from=mimotwo]{north west}
  \inheritanchor[from=mimotwo]{south west}
  \inheritanchor[from=mimotwo]{south east}
  \inheritanchor[from=mimotwo]{text}
  \inheritanchor[from=mimotwo]{A}
  \inheritanchor[from=mimotwo]{B}
  \inheritanchor[from=mimotwo]{first antenna base start}
  \inheritanchor[from=mimotwo]{second antenna base start}
  \inheritanchor[from=mimotwo]{first antenna base end}
  \inheritanchor[from=mimotwo]{second antenna base end}
  \inheritanchor[from=mimotwo]{first antenna}
  \inheritanchor[from=mimotwo]{second antenna}
  
  \anchor{C}{\anchorC}

  \anchor{third antenna base start}{
    \pgf@process{\anchorC}
  }
  \anchor{third antenna base end}{
    \pgf@process{\anchorC}
    \advance\pgf@x by \antennaoffset%
    \advance\pgf@y by 0.2cm%
  }
  \anchor{third antenna}{
    \pgf@process{\anchorC}
    \advance\pgf@x by \antennaoffset%
    \advance\pgf@y by 0.4cm%
  }

  \backgroundpath{
    \pgfpathrectanglecorners
    { 
      \pgfpointadd{\southwest}{\pgfpoint{\pgfkeysvalueof{/pgf/outer xsep}}{\pgfkeysvalueof{/pgf/outer ysep}}}
    }
    { 
      \pgfpointadd{\northeast}{\pgfpointscale{-1}{\pgfpoint{\pgfkeysvalueof{/pgf/outer xsep}}{\pgfkeysvalueof{/pgf/outer ysep}}}}
    }

    \pgf@process{\anchorA}
    \mimoshapesdrawantenna

    \pgf@process{\anchorB}
    \mimoshapesdrawantenna

    \pgf@process{\anchorC}
    \mimoshapesdrawantenna
  }
}

\pgfdeclareshape{mimoind}{

  \inheritsavedanchors[from=mimothree]

  \inheritanchorborder[from=mimothree]

  \inheritanchor[from=mimothree]{center}
  \inheritanchor[from=mimothree]{north}
  \inheritanchor[from=mimothree]{east}
  \inheritanchor[from=mimothree]{south}
  \inheritanchor[from=mimothree]{west}
  \inheritanchor[from=mimothree]{north east}
  \inheritanchor[from=mimothree]{north west}
  \inheritanchor[from=mimothree]{south west}
  \inheritanchor[from=mimothree]{south east}
  \inheritanchor[from=mimothree]{text}
  \inheritanchor[from=mimothree]{A}
  \inheritanchor[from=mimothree]{B}
  \inheritanchor[from=mimothree]{C}
  \inheritanchor[from=mimothree]{first antenna base start}
  \inheritanchor[from=mimothree]{second antenna base start}
  \inheritanchor[from=mimothree]{third antenna base start}
  \inheritanchor[from=mimothree]{first antenna base end}
  \inheritanchor[from=mimothree]{second antenna base end}
  \inheritanchor[from=mimothree]{third antenna base end}
  \inheritanchor[from=mimothree]{first antenna}
  \inheritanchor[from=mimothree]{second antenna}
  \inheritanchor[from=mimothree]{third antenna}

  \backgroundpath{
    \pgfpathrectanglecorners
    { 
      \pgfpointadd{\southwest}{\pgfpoint{\pgfkeysvalueof{/pgf/outer xsep}}{\pgfkeysvalueof{/pgf/outer ysep}}}
    }
    { 
      \pgfpointadd{\northeast}{\pgfpointscale{-1}{\pgfpoint{\pgfkeysvalueof{/pgf/outer xsep}}{\pgfkeysvalueof{/pgf/outer ysep}}}}
    }

    \pgf@process{\anchorA}
    \mimoshapesdrawantenna

    \pgf@process{\anchorB}
    \drawellipses
    
    \pgf@process{\anchorC}
    \mimoshapesdrawantenna
  }
}

\pgfdeclareshape{boxtwo}{

  \inheritsavedanchors[from=mimoone]

  \inheritanchorborder[from=mimoone]

  \inheritanchor[from=mimotwo]{center}
  \inheritanchor[from=mimotwo]{north}
  \inheritanchor[from=mimotwo]{east}
  \inheritanchor[from=mimotwo]{south}
  \inheritanchor[from=mimotwo]{west}
  \inheritanchor[from=mimotwo]{north east}
  \inheritanchor[from=mimotwo]{north west}
  \inheritanchor[from=mimotwo]{south west}
  \inheritanchor[from=mimotwo]{south east}
  \inheritanchor[from=mimotwo]{text}

  \anchor{A}{
    \pgf@process{\northeast}%
    \pgf@x=0pt
    \pgf@y=0.5\pgf@y%
  }

  \anchor{B}{
    \pgf@process{\northeast}%
    \pgf@x=0pt
    \pgf@y=-0.5\pgf@y%
  }

  \backgroundpath{
    {
      \pgfpathrectanglecorners
      { 
        \pgfpointadd{\southwest}{\pgfpoint{\pgfkeysvalueof{/pgf/outer xsep}}{\pgfkeysvalueof{/pgf/outer ysep}}}
      }
      { 
        \pgfpointadd{\northeast}{\pgfpointscale{-1}{\pgfpoint{\pgfkeysvalueof{/pgf/outer xsep}}{\pgfkeysvalueof{/pgf/outer ysep}}}}
      }
      \pgfusepath{stroke,fill}
    }
    
    {
      \csname pgf@anchor@\pgf@sm@shape@name @A\endcsname
      \pgfcircle{\pgfpoint{\pgf@x}\pgf@y{}}{\pgfkeysvalueof{/pgf/circle radius}}
      \csname pgf@anchor@\pgf@sm@shape@name @B\endcsname
      \pgfcircle{\pgfpoint{\pgf@x}\pgf@y{}}{\pgfkeysvalueof{/pgf/circle radius}}
      \pgfsetfillcolor{\pgfkeysvalueof{/pgf/circle fill color}}
      \pgfsetstrokecolor{\pgfkeysvalueof{/pgf/circle stroke color}}
      \pgfusepath{stroke,fill}
    }
  }
}

\pgfdeclareshape{boxthree}{

  \inheritsavedanchors[from=boxtwo]

  \inheritanchorborder[from=boxtwo]

  \inheritanchor[from=boxtwo]{center}
  \inheritanchor[from=boxtwo]{north}
  \inheritanchor[from=boxtwo]{east}
  \inheritanchor[from=boxtwo]{south}
  \inheritanchor[from=boxtwo]{west}
  \inheritanchor[from=boxtwo]{north east}
  \inheritanchor[from=boxtwo]{north west}
  \inheritanchor[from=boxtwo]{south west}
  \inheritanchor[from=boxtwo]{south east}
  \inheritanchor[from=boxtwo]{text}

  \anchor{A}{
    \northeast
    \pgf@x=0pt
    \pgf@y=0.6\pgf@y%
  }

  \anchor{B}{
    \pgf@x=0pt
    \pgf@y=0pt
  }

  \anchor{C}{
    \northeast
    \pgf@x=0pt
    \pgf@y=-0.6\pgf@y%
  }

  \backgroundpath{
    {
      \pgfpathrectanglecorners
      { 
        \pgfpointadd{\southwest}{\pgfpoint{\pgfkeysvalueof{/pgf/outer xsep}}{\pgfkeysvalueof{/pgf/outer ysep}}}
      }
      { 
        \pgfpointadd{\northeast}{\pgfpointscale{-1}{\pgfpoint{\pgfkeysvalueof{/pgf/outer xsep}}{\pgfkeysvalueof{/pgf/outer ysep}}}}
      }
      \pgfusepath{stroke,fill}
    }
    
    {
      \csname pgf@anchor@\pgf@sm@shape@name @A\endcsname
      \pgfcircle{\pgfpoint{\pgf@x}\pgf@y{}}{\pgfkeysvalueof{/pgf/circle radius}}
      \csname pgf@anchor@\pgf@sm@shape@name @B\endcsname
      \pgfcircle{\pgfpoint{\pgf@x}\pgf@y{}}{\pgfkeysvalueof{/pgf/circle radius}}
      \csname pgf@anchor@\pgf@sm@shape@name @C\endcsname
      \pgfcircle{\pgfpoint{\pgf@x}\pgf@y{}}{\pgfkeysvalueof{/pgf/circle radius}}
      \pgfsetfillcolor{\pgfkeysvalueof{/pgf/circle fill color}}
      \pgfsetstrokecolor{\pgfkeysvalueof{/pgf/circle stroke color}}
      \pgfusepath{stroke,fill}
      }
  }
}

\makeatother

\DeclareMathOperator{\diag}{diag}
\DeclareMathOperator{\supp}{supp}
\DeclareMathOperator{\trace}{tr}
\DeclareMathOperator{\rank}{rank}

\begin{document}

\title{Bayesian Algorithms for Kronecker-structured Sparse Vector Recovery With Application to IRS-MIMO Channel Estimation}

\author{Yanbin He and Geethu Joseph
\thanks{The material in this paper was presented
in part at the IEEE International Conference on Acoustics, Speech, \& Signal Processing (ICASSP), June 2023, Rhodes, Greece~\cite{he2022structure}.

The authors are with the Signal Processing Systems group, Electrical Engineering, Mathematics, and Computer Science faculty, at the Delft University of Technology, The Netherlands. Emails:$\{\text{y.he-1, g.joseph}\}\text{@tudelft.nl}$.}
}

\markboth{Journal of \LaTeX\ Class Files,~Vol.~14, No.~8, August~2021}%
{Shell \MakeLowercase{\textit{et al.}}: A Sample Article Using IEEEtran.cls for IEEE Journals}

\IEEEpubid{0000--0000/00\$00.00~\copyright~2021 IEEE}

\maketitle

\begin{abstract}
We study the sparse recovery problem with an underdetermined linear system characterized by a Kronecker-structured dictionary and a Kronecker-supported sparse vector. We cast this problem into the sparse Bayesian learning (SBL) framework and rely on the expectation-maximization method for a solution. To this end, we model the Kronecker-structured support with a hierarchical Gaussian prior distribution parameterized by a Kronecker-structured hyperparameter, leading to a non-convex optimization problem. The optimization problem is solved using the alternating minimization (AM) method and a singular value decomposition (SVD)-based method, resulting in two algorithms. Further, we analytically guarantee that the AM-based method converges to the stationary point of the SBL cost function. The SVD-based method, though it adopts approximations, is empirically shown to be more efficient and accurate. We then apply our algorithm to estimate the uplink wireless channel in an intelligent reflecting surface-aided MIMO system and extend the AM-based algorithm to address block sparsity in the channel. We also study the SBL cost to show that the minima of the cost function are achieved at sparse solutions and that incorporating the Kronecker structure reduces the number of local minima of the SBL cost function. Our numerical results demonstrate the effectiveness of our algorithms compared to the state-of-the-art.
\end{abstract}


\section{Introduction}\label{sec:intro}
\IEEEPARstart{T}{he} basis expansion model (BEM) is widely used in various fields, such as image processing \cite{he2009exploiting,he2009tree,ma2003maximum,barhumi2005mmse} and wireless communications \cite{giannakis1998basis,han2018low,zheng2022survey}, to obtain flexible linear representations for non-linear functions. It is achieved by approximating non-linear functions as linear combinations of simple \emph{basis functions} \cite{hastie2009elements}. BEM comprises a linear model with a \emph{dictionary} of basis functions and \emph{coefficients} associated with the dictionary. This paradigm, united with compressed sensing (CS), can be employed to estimate the unknown parameters of non-linear functions. The idea is to sample the unknown parameter over a range to construct known functions that comprise an over-complete dictionary. Here, only a few basis functions corresponding to the true parameters are activated, resulting in sparse coefficients. Thus, CS techniques can be leveraged to reconstruct the sparse coefficients from the linear model. Further, many signal representation problems in wireless communications \cite{araujo2019tensor,chang2021sparse,xu2022sparse,he2022structure} and image processing \cite{zhao2019exploiting,yang2015compressive} applications use multidimensional BEM \cite{caiafa2012block,caiafa2013computing,duarte2010kronecker,duarte2011kronecker,caiafa2013multidimensional}. These models can be represented using a sparse vector with Kronecker-structured support, i.e., the support of the sparse vector is the Kronecker product of several low-dimensional support vectors. This work investigates the CS problem of solving a high-dimensional underdetermined linear system of equations to find a Kronecker-structured sparse solution.


The multidimensional basis expansion models are also typically associated with Kronecker-structured dictionaries, i.e., the Kronecker product of multiple dictionary matrices. For example, consider the channel estimation problem for the intelligent reflecting surfaces (IRS)-aided wireless communication systems \cite{he2022structure}. The BEM for the unknown channel matrix is constructed by sampling pre-defined spatial angle grids and forming a dictionary using the corresponding steering vectors. Then, the channel matrix can be represented using a three-dimensional sparse BEM coefficient vector where the three dimensions (mode) are the angle-of-departure (AoD), angle-of-arrival (AoA), and difference of the AoA and AoD at the IRS. Furthermore, different combinations of AoDs and AoAs naturally elicit the Kronecker product, leading to a Kronecker-structured dictionary and sparse coefficients. Motivated by such multi-parameter estimation problems, we consider the following Kronecker-structured linear inversion problem,\IEEEpubidadjcol
\begin{equation}\label{eq.problem_basic}
    \bm y = \bm H \bm x+\bm n,
\end{equation}
where $\bm y \in \mathbb{C}^{\bar{M} \times 1}$ is the noisy measurement, $\bm H \in \mathbb{C}^{\bar{M} \times \bar{N}}$ is the Kronecker-structured dictionary with $\bar{M} < \bar{N}$, $\bm x \in \mathbb{C}^{\bar{N} \times 1}$ is the unknown sparse BEM coefficient vector, and $\bm n$ is the measurement noise. Specifically,
\begin{equation}\label{eq.separable_dict}
    \bm H = \otimes_{i=1}^I \bm H_i,
\end{equation}
where $\bm H_i \in \mathbb{C}^{M_i \times N_i}$ with $\prod_{i=1}^I M_i = \bar{M}$, and $\prod_{i=1}^I N_i = \bar{N}$. The sparse vector $\bm x$ possesses a Kronecker-structured support vector given by $\otimes_{i=1}^I \bm b_i\in\{0,1\}^{\bar{N}\times 1}$ where $\bm b_i \in \{0,1\}^{N_i \times 1}$ is a binary (support) vector. We aim to reconstruct the sparse vector $\bm x$ given the Kronecker-structured dictionary $\bm H$, noisy measurement $\bm y$ by exploiting its 
Kronecker-structured support.

Various signal processing techniques have been proposed to reconstruct the sparse vector $\bm x$ from the linear measurement $\bm y$ and Kronecker-structured dictionary $\bm H$. \cite{caiafa2013computing} designed a greedy method called Kronecker-orthogonal matching pursuit (KOMP) to generalize the traditional OMP for multidimensional signals. 
Like OMP, the KOMP algorithm has low complexity but requires hand-tuning of a sensitive stopping threshold. Recently, parameter tuning-free approaches based on sparse Bayesian learning (SBL)  were studied \cite{chang2021sparse}. SBL is known to have superior performance~\cite{wipf2004sparse} and flexibility to incorporate several additional structures along with sparsity \cite{fang2014pattern,wang2018alternative,prasanna2021mmwave,wu2022clustered}. The Kronecker structure can also be incorporated into the SBL framework. \cite{zhao2019exploiting} introduced the SBL approach for the Kronecker structure in~\eqref{eq.separable_dict} with $I=3$, using tensor-wise hyperparameters. The linear inversion problem with $I=2$ was considered in \cite{bualtoiu2021sparse}. This framework was later generalized to the $I$-dimensional tensor and applied to wireless communications \cite{chang2021sparse,xu2022sparse}. However, the derivation of the SBL algorithm in \cite{chang2021sparse,xu2022sparse} relied on several approximations leading to a suboptimal recovery accuracy. Hence, we seek novel Bayesian algorithms exploiting the Kronecker and sparse structures to improve reconstruction accuracy and efficiency.

Furthermore, the existing Bayesian algorithms exploiting the Kronecker structure \cite{chang2021sparse,xu2022sparse} lack theoretical guarantees on their convergence and performance, and these properties are only demonstrated empirically. 
The study on convergence is limited because the convergence guarantee of the classic SBL using the expectation-maximization (EM) algorithm cannot be trivially extended for these algorithms. \cite{xu2022sparse,chang2021sparse} claimed that the solution to the underlying optimization problem should be obtained at the stationary point of the cost function. However, it is unclear whether the stationary point can be reached due to its iterative nature and approximations. 
Similarly, the performance improvement due to incorporating the Kronecker structure into the SBL has been shown in \cite{chang2021sparse} without any theoretical justification.
Given the shallow theoretical analysis of the existing algorithms and analysis of the SBL cost function, we make progress on these problems by analyzing our new Bayesian algorithms.

\emph{Our contributions:} We devote this paper to the algorithm development and convergence analysis of algorithms for the Kronecker-structured sparse recovery problem in~\eqref{eq.problem_basic}.
Our main contributions are as follows:
\begin{itemize}[leftmargin=*]
    \item \emph{Algorithm development:} We present our new Bayesian recovery algorithms in Sec.~\ref{sec:sbl}. We first present two novel SBL algorithms for Kronecker-structured sparse recovery called KroSBL. The first KroSBL algorithm, which is based on alternating minimization (AM), solves the underlying optimization problem of the SBL algorithm using the AM procedure. The second KroSBL algorithm, based on singular value decomposition (SVD),  is faster and uses a simple approximation to obtain the SBL algorithm. 

    \item \emph{Application:} We apply our problem to a prototypical application of IRS-aided MIMO channel estimation in Sec.~\ref{sec:channelesti}. Besides the Kronecker-structure sparsity, the BEM representation of IRS-cascaded channels also exhibits block sparsity where the nonzero entries occur in clusters. To handle this additional structure, we extend our AM-based algorithm for Kronecker-structured block sparsity based on non-negative least squares.
    \item \emph{Convergence guarantee:} We derive convergence guarantees for the AM-based algorithm in Sec.~\ref{sec.con_AM}. We establish that \emph{i)} the AM procedure can attain the stationary point of the cost function in the M-step, and $ii)$ the AM-based algorithm is guaranteed to converge to the stationary point of the SBL cost function in the noisy setting.
    
    \item \emph{Cost function analysis:} We examine the local minima of the KroSBL cost function in Sec.~\ref{sec:local_minima}. Assuming the sparse vector in~\eqref{eq.problem_basic} to satisfy $\bm x = \otimes_{i=1}^I \bm x_i$, we prove that all local minima are sparse in the noiseless case. Besides, we demonstrate that incorporating the Kronecker structure in the hyperparameter vector as the KroSBL cost function can significantly reduce the local minima under the unique representation property (URP) assumption for $\bm H$.

    \item \emph{Numerical Results:} We assess the our schemes in two scenarios in Sec.~\ref{sec:numsimu}. Firstly, we study the sparse recovery performance of the presented schemes against algorithms in the literature to demonstrate its superior recovery accuracy and run time. Secondly, we conduct a case study on channel estimation for IRS-aided systems with our approach and illustrate that incorporating block sparsity in our scheme can further enhance performance.
    
\end{itemize}

Overall, we present two algorithms for sparse signal recovery that arises in BEM with multiple unknown parameters. 
The first algorithm, AM-KroSBL, enjoys solid theoretical guarantees, while the other algorithm, SVD-KroSBL, is computationally light and practically more relevant. 

\emph{Notation:} Boldface small letters denote vectors, and boldface capital letters denote matrices. The symbols $[\bm x]_i$, $[\bm X]_{i}$, and $[\bm X]_{ij}$ represent the $i$-th entry of vector $\bm x$, the $i$-th column of matrix $\bm X$, and the entry on the $i$-th row and $j$-th column of matrix $\bm X$, respectively. We denote the all-one vector with length $N$ as $\bm 1_N$. The symbol $\|\cdot\|_p$ denotes the vector $\ell_p$ norm. We use $\bm x > 0$ (or $\bm x \geq$) to denote that all the entries in $\bm x$ are positive (or non-negative). 
Depending on the argument, the vector element-wise inversion or matrix inversion is denoted as $(\cdot)^{-1}$. If the argument is a vector, operator $\diag(\cdot)$ returns a diagonal matrix with the argument along the diagonal, and it returns a vector of its diagonal entries if the argument is a matrix. The symbols $(\cdot)^\mathsf{T}$, $(\cdot)^*$, $(\cdot)^\mathsf{H}$, $|\cdot|$, and $(\cdot)^\dagger$ are the matrix operations of transpose, conjugate, conjugate transpose, determinant, and pseudo-inverse, respectively. Also, $\otimes$ and $\odot$ represent the Kronecker and Khatri-Rao products, respectively. The matrix $\bm I_N$ denotes the identity matrix of size $N\times N$. We use $\mathcal{CN}(\bm a,\bm B)$ to denote the complex Gaussian distribution with mean $\bm a$ and covariance $\bm B$. The set of real, complex matrices of size $M\times N$ is represented by $\mathbb{R}^{M\times N}$ and $\mathbb{C}^{M\times N}$, respectively. 


\section{Kronecker-structured Sparse Bayesian Learning}
\label{sec:sbl}

We study the model in~\eqref{eq.problem_basic} with additive Gaussian noise $\bm n$ following $\mathcal{CN}(\bm 0,\sigma^2 \bm I_{\bar{M}})$. For simplicity, we assume that the noise variance $\sigma^2$ is known and $N_i = N$ for $i=1,2,\ldots,I$. This section presents new recovery algorithms to solve for $\bm x$ from \eqref{eq.problem_basic}, exploiting the Kronecker structure.

Inspired by the SBL framework~\cite{wipf2004sparse}, we impose a fictitious sparsity promoting zero-mean Gaussian prior on $\bm x$ with an unknown covariance matrix $\bm\Gamma\in\mathbb{R}^{N^I\times N^I}$. We construct the covariance matrix as $\bm \Gamma = \diag(\bm \gamma)$ with $\bm \gamma = \otimes_{i=1}^I \bm \gamma_i$ and $\bm \gamma_i\in\mathbb{R}^{N\times 1}$, mimicking the Kronecker-structured support. Specifically,
\begin{equation}\label{eq.sparseprior}
p(\bm x;\{\bm \gamma_i\})=\mathcal{CN}(\bm 0,\bm \Gamma),
\end{equation}
where $\{\bm \gamma_i\}$ is the simplified notation for $\{\bm \gamma_i\}_{i=1}^I$ used henceforth. Then, we turn to the type-II ML estimation, i.e., we first estimate the hyperparameters $\{\bm \gamma_i\}$, based on which the MAP estimate of $\bm x$ obtained as $\arg \max_{\bm x} p(\bm x|\bm y;\{\bm \gamma_i\})$ \cite{kreutz2024dictionaries}. The ML estimates of $\{\bm \gamma_i\}$ are obtained by minimizing the negative log-likelihood, i.e., the KroSBL cost function is given by
\begin{equation}\label{eq.cost_ml}
    \mathcal{L}\left(\{\bm\gamma_i\}\right)=-\log p(\bm y;\{\bm \gamma_i\}) = \log |\bm \Sigma_{\bm y}| + \bm y^\mathsf{H} \bm \Sigma_{\bm y}^{-1} \bm y,
\end{equation}
where $\bm \Sigma_{\bm y} = \sigma^2 \bm I_{\bar{M}} + \bm H \bm \Gamma \bm H^\mathsf{H}$ \cite{wipf2004sparse}. We note that $\bm \gamma = \otimes_{i=1}^I \bm \gamma_i=\otimes_{i=1}^I \alpha_i \bm \gamma_i$ for any $\alpha_i\!>\!0$ when $\prod_{i=1}^I \alpha_i = 1$. Thus, if $\{\bm \gamma_i\}$ maximizes~\eqref{eq.cost_ml}, then $\{\alpha_i\bm \gamma_i\}$ also achieves the maximum for any $\alpha_i>0$ with $\prod_i\alpha_i=1$. To eliminate this scaling ambiguity, we normalize the hyperparameter vectors, i.e., we impose the constraint $\|\bm \gamma_i\|_2 = 1$ for $i=1,2,\ldots,I-1$. So the ML problem to estimate $\{\gamma_i\}$ reduces to
\begin{equation}\label{eq.ml_problem}
    \underset{{\{\bm \gamma_i\}\in \mathcal{C}}}{\min} \mathcal{L}\left(\{\bm\gamma_i\}\right),
\end{equation}
where we define the constraint set
\begin{equation}
    \mathcal{C}\!=\!\left\{\{\bm \gamma_i\}\bigg|\bm \gamma_i \geq 0,i=1,2,\ldots,I,\|\bm \gamma_i\|_2=1, \forall i\neq I \right\}.
\end{equation}

The problem in \eqref{eq.ml_problem} does not admit a closed-form solution. Thus, we resort to the standard EM algorithm~\cite{wipf2004sparse} for an iterative solution. Specifically, the $r$th iteration of~EM~is
\begin{align}
&\text{\bf E-step:}~\text{Compute }\mathbb{E}_{\bm x|\bm y;\bm \gamma^{(r)}}\{\log[p(\bm y,\bm x;\{\bm \gamma_i\})]\},\label{eq.estep}\\
&\text{\bf M-step:}~\{\bm \gamma_i\}^{(r+1)} = \underset{\substack{\{\bm \gamma_i\} \in \mathcal{C}_+}}{\arg\max} ~\mathbb{E}_{\bm x|\bm y;\bm \gamma^{(r)}}\{\log[p(\bm y,\bm x;\{\bm \gamma_i\})]\},\label{eq.mstep_basic}  
\end{align}
where $\bm\gamma^{(r)} = \otimes_{i=1}^I \bm \gamma_i^{(r)}$ is the estimate in the $r$th iteration and
\begin{equation}
\mathcal{C}_+\!=\!\left\{\{\bm \gamma_i\}\in\mathcal{C}\bigg|\bm \gamma_i >0,i=1,2,\!\ldots\!,I\right\}.
\end{equation}
Here, we restrict the feasible set in~\eqref{eq.mstep_basic} to $\mathcal{C}_+$ instead of $\mathcal{C}$ in~\eqref{eq.ml_problem} to avoid degenerate distributions. Further, using straightforward algebraic simplifications \cite{bishop2006pattern}, we can reduce~\eqref{eq.mstep_basic} to
\begin{equation}\label{eq.qfunc}
\{\bm \gamma_i\}^{(r+1)}
= \underset{\substack{\{\bm \gamma_i\} \in \mathcal{C}_+}}{\arg\min} 
\log\!|\!\diag(\bm\gamma)|\!+\!(\bm d^{(r)})^{\mathsf{T}}\bm\gamma^{-1},
\end{equation}    
where we define \begin{equation}\label{eq.compute_d}
\bm d^{(r)} = \diag(\bm \Sigma_{\bm x}+\bm \mu_{\bm x}\bm \mu_{\bm x}^\mathsf{H}).
\end{equation}
Here, $\bm \mu_{\bm x}$ and $\bm \Sigma_{\bm x}$, which depend on $\bm \gamma^{(r)}$, are the mean and variance of conditional distribution $p(\bm x|\bm y;\bm \gamma^{(r)})$, respectively,
\begin{equation}
\label{eq.post_meva}
\bm \mu_{\bm x} = \sigma^{-2}\bm \Sigma_{\bm x}\bm H^\mathsf{H}\bm y,\ \bm \Sigma_{\bm x} = \!\left[\sigma^{-2}\bm H^\mathsf{H}\bm H\!+\!\diag(\bm \gamma^{(r)})^{-1}\right]^{-1}\!. 
\end{equation}

\begin{algorithm}[t!]
  \SetAlgoLined
  \DontPrintSemicolon
  \KwData{Measurements $\bm y$, matrix $\bm H$, noise power $\sigma^2$}
  \textbf{Parameters}: Threshold $\epsilon$ and $\epsilon_{\mathsf{AM}}$
   
  \textbf{Initialization}: $\{\bm \gamma_i\}^{(0)}=\bm 0$, $\{\bm \gamma_i\}^{(1)} \in \mathcal{C}$, set $r=1$
  
   \While{$\| \otimes_{i=1}^I \bm \gamma_i^{(r)}-\otimes_{i=1}^I \bm \gamma_i^{(r-1)} \|_2 > \epsilon$}{
    
    Compute $\bm d^{(r)}$ using~\eqref{eq.compute_d} and~\eqref{eq.post_meva}
    
    Set $t=1$, $\{\bm \gamma_i\}^{(r,1)} = \{\bm \gamma_i\}^{(r)}$, $\{\bm \gamma_i\}^{(r,0)} = \bm 0$ 
    
    \While{$\| \otimes_{i=1}^I \bm \gamma_i^{(r,t)}-\otimes_{i=1}^I \bm \gamma_i^{(r,t-1)} \|_2 > \epsilon_\mathsf{AM}$}{

    
    Compute $\{\bm \gamma_i\}^{(r,t+1)}$ using \eqref{eq.update} and \eqref{eq.projection}
    

    Update AM iteration index $t \leftarrow t + 1$

  }

    $\{\bm \gamma_i\}^{(r+1)} = \{\bm \gamma_i\}^{(r,t)}$
  
    Update iteration index $r \leftarrow r + 1$

   } 
   
   \KwResult{Output $\bm x=\bm \mu_{\bm x}$ using~\eqref{eq.post_meva}}
  \caption{AM-KroSBL}
  \label{al.AMKroSBL}
\end{algorithm}

Let $Q(\{\bm \gamma_i\}|\bm \gamma^{(r)})=\log|\diag(\bm\gamma)|+(\bm d^{(r)})^{\mathsf{T}}\bm\gamma^{-1}$, which depends on $\bm \gamma^{(r)}$ via $\bm d^{(r)}$. Then, the M-step is written~as
\begin{equation}\label{eq.mstep}
\underset{\substack{\{\bm \gamma_i\}}}{\min}\ Q(\{\bm \gamma_i\}|\bm \gamma^{(r)}) \hspace{0.3cm} \text{s.t.}\
\bm\gamma = \otimes_{i=1}^I \bm \gamma_i,\ \{\bm \gamma_i\}\in \mathcal{C}_+.
\end{equation}
The solution to~\eqref{eq.mstep} without the constraint $\bm\gamma = \otimes_{i=1}^I \bm \gamma_i$ is straightforward \cite{wipf2004sparse}. However, the Kronecker constraint poses a challenge to derive a closed-form solution. To solve~\eqref{eq.mstep} with the Kronecker constraint, two distinct ways are presented: AM-based and SVD-based approaches.

\subsection{AM-based KroSBL (AM-KroSBL)}

AM-KroSBL solves \eqref{eq.mstep} by alternatingly updating one hyperparameter vector while keeping the others fixed. We first compute the gradient of cost function\footnote{We interchangeably use $Q(\{\bm \gamma_i\}|\bm \gamma^{(r)})$, $Q(\{\bm \gamma_i\})$, and $Q$ in the paper.} 
$Q(\{\bm \gamma_i\})$ with respect to $\{\bm \gamma_i\}$ and set it to zero, leading to
\begin{equation}
\label{eq.update}
    \tilde{\bm \gamma}_i\!=\!N^{-I+1}\!\big[(\otimes_{j=1}^{i-1}(\tilde{\bm \gamma}_j)^{-1})\otimes \bm I_N \otimes (\otimes_{j=i+1}^I(\bm \gamma_j^{(r,t)})^{-1})\big]^\mathsf{T}\bm d^{(r)},
\end{equation}
for $i=1,2,\ldots,I$, with $\bm\gamma_l^{(r,t)}$ is the estimate in the $t$th iteration of AM and the $r$th iteration of EM. To avoid the scaling ambiguity, in each iteration $t$, we project the hyperparameter vectors to $\mathcal{C}_+$
as
\begin{equation}\label{eq.projection}
\bm \gamma_i^{(r,t+1)} = \begin{cases}
\frac{\tilde{\bm \gamma}_i}{\|\tilde{\bm \gamma}_i\|_2} &\text{for}\; i=1,2,\ldots,I-1\\
\prod_{j=1}^{I-1}\|\tilde{\bm \gamma}_i\|_2\tilde{\bm \gamma}_I &\text{for}\; i=I.
\end{cases}
\end{equation}
The projection~\eqref{eq.projection} does not change the cost function value $Q$ as $\otimes_{i=1}^I\bm \gamma_i^{(r,t+1)}=\otimes_{i=1}^I\tilde{\bm \gamma_i}$. The steps are summarized in Algorithm \ref{al.AMKroSBL}. The AM-KroSBL is guaranteed to improve the cost function given by~\eqref{eq.qfunc} after every iteration. But due to its iterative nature, it is computationally inefficient. Thus, we next present a non-iterative method to solve~\eqref{eq.mstep} based on SVD.

\subsection{SVD-based KroSBL (SVD-KroSBL)}

This method solves \eqref{eq.mstep} without the constraint $\bm\gamma = \otimes_{i=1}^I \bm \gamma_i$ first and then projects the solution to the constraint set. We note from \cite{wipf2004sparse} that
\begin{equation}\label{eq.mstep_svd}
\underset{\substack{\bm\gamma }}{\arg\min}\log|\diag(\bm\gamma)|+(\bm d^{(r)})^\mathsf{T}\bm\gamma^{-1} = \bm d^{(r)}.
\end{equation}
To project the above solution to the constraint set, we solve for $\{\bm \gamma_i\}$ that minimizes $\| \bm d^{(r)} - \otimes_{i=1}^I \bm \gamma_i \|_2$.
We further approximate this optimization problem as $(I-1)$ rank-1 approximations:
\begin{equation}\label{prob.bidecom}
\bm\gamma_{i}^{(r+1)} =\underset{\bm\gamma_{i}: \|\bm \gamma_{i}\|_2=1,
\bar{\bm\gamma}_{i}\in\mathbb{R}^{N(I-i)}}{\arg\min} \| \bar{\bm\gamma}_{i-1} - \bm \gamma_{i}\otimes\bar{\bm\gamma}_{i} \|_2,\ \forall i=1,2,\ldots,I-1,   
\end{equation}
where $\bar{\bm \gamma}_0=\bm d^{(r)}$ and $\bar{\bm \gamma}_{I-1}=\bm \gamma_I$. The problem in \eqref{prob.bidecom} is solved using SVD. This approach is summarized in Algorithm \ref{al.SVDKroSBL}.

\begin{algorithm}[t]
  \SetAlgoLined
  \DontPrintSemicolon
\SetNoFillComment
  \KwData{Measurements $\bm y$, matrix $\bm H$, noise power $\sigma^2$}
  \textbf{Parameters}: Threshold $\epsilon$
   
  \textbf{Initialization}: $\{\bm \gamma_i\}^{(0)}=\bm 0$, $\{\bm \gamma_i\}^{(1)} \in \mathcal{C}$, set $r=1$
  
   \While{$\| \otimes_{i=1}^I \bm \gamma_i^{(r)}-\otimes_{i=1}^I \bm \gamma_i^{(r-1)} \|_2 > \epsilon$}{
    
    Compute $\bm d^{(r)}$ using~\eqref{eq.compute_d} and~\eqref{eq.post_meva} 
    
    \For{$i=1,\ldots,I-1$} 
    { Solve~\eqref{prob.bidecom} for $\bm \gamma_i^{(r+1)}$         
    }

    
    Update iteration index: $r \leftarrow r + 1$

   }
   
   \KwResult{Output $\bm x=\bm \mu_{\bm x}$ using~\eqref{eq.post_meva}}
  \caption{SVD-KroSBL}
  \label{al.SVDKroSBL}
\end{algorithm}


\section{Application and Extension}\label{sec:channelesti}
In this section, we discuss the application of our algorithm to channel estimation in an IRS-assisted MIMO system. 

\subsection{Cascaded channel estimation for IRS-aided MIMO}
Consider an uplink MIMO millimeter-wave/terahertz band system with a $T$-antenna transmitter mobile station (MS), an $R$-antenna receiver base station (BS), and an $L$-element uniform linear array IRS. 
Let $\bm H_\mathrm{MS}\in \mathbb{C}^{L\times T}$ and $\bm H_\mathrm{BS} \in \mathbb{C}^{R\times L}$ denote the MS-IRS and IRS-BS (narrowband) channels, respectively, 
\begin{align}
\label{eq.channelmodel1}
\bm H_\mathrm{MS}&=\sum_{p=1}^{P_{\mathrm{MS}}}\sqrt{\frac{LT}{P_{\mathrm{MS}}}}\beta_{{\mathrm{MS}},p}\bm a_L(\phi_{{\mathrm{MS}},p})\bm a_{T}(\alpha_{{\mathrm{MS}}})^\mathsf{H}\\\label{eq.channelmodel2}
\bm H_\mathrm{BS}&=\sum_{p=1}^{P_{\mathrm{BS}}}\sqrt{\frac{RL}{P_{\mathrm{BS}}}}\beta_{{\mathrm{BS}},p}\bm a_{R}(\alpha_{{\mathrm{BS}},p})\bm a_L(\phi_{{\mathrm{BS}}})^\mathsf{H},
\end{align}    
where $P_{\mathrm{MS}}$ and $P_{\mathrm{BS}}$ are the number of rays. Also, for any integer $Q$ and angle $\psi$, steering vector $\bm a_Q(\psi)\in\mathbb{C}^{Q\times 1}$ with half-wavelength spacing is
\begin{equation}\label{eq.steer}
\bm a_Q(\psi) = \frac{1}{\sqrt{Q}}\begin{bmatrix}
1 & e^{j\pi  \cos{\psi}}\cdots e^{j\pi (Q-1)  \cos{\psi}}\end{bmatrix}^\mathsf{T}.
\end{equation}
The angles $\phi_{{\mathrm{MS}},p}$, $\alpha_{\mathrm{MS}}$, $\alpha_{{\mathrm{BS}},p}$, and $\phi_{{\mathrm{BS}}}$ denote the $p$th AoA of the IRS, and the $p$th AoD of the MS, the $p$th AoA of the BS, and the AoD of the IRS, respectively (see Fig.~\ref{fig:irs_sys}). Then, the cascaded MS-IRS-BS channel is then given by $\bm H_\mathrm{BS}\diag (\bm \theta)\bm H_\mathrm{MS}$ for a given IRS configuration $\bm\theta \in \mathbb{C}^{L \times 1}$. Here, the $i$th entry of $\bm\theta$ represents the gain and phase shift due to the $i$th IRS element. 
We aim to estimate the cascaded channel $\bm H_\mathrm{BS}\diag (\bm \theta)\bm H_\mathrm{MS}$ for any $\bm \theta$. 

To estimate the channel, we send pilot symbols over $K$ time slots over which $\bm H_\mathrm{MS}$ and $\bm H_\mathrm{BS}$ is assumed to be constant. We choose $K_{\mathrm{I}}<K$ IRS configurations, and for each configuration, we transmit pilot data $\bm X \in \mathbb{C}^{T\times K_{\mathrm{P}}}$ over $K_{\mathrm{P}}$ time slots such that $K=K_{\mathrm{I}}K_{\mathrm{P}}$. Hence, the received signal $\bm Y_k\in \mathbb{C}^{R \times K_{\mathrm{P}}}$ corresponding to the $k$th configuration $\bm \theta_k$~is
\begin{equation}\label{eq.basicdatamodel}
\bm Y_k=\bm H_\mathrm{BS}\diag (\bm \theta_k)\bm H_\mathrm{MS} \bm X+\bm W_k,
\end{equation}
where $\bm W_k \in \mathbb{C}^{R\times K_{\mathrm{P}}}$ is the additive white Gaussian noise with zero mean and variance $\sigma^2$. 
To estimate the cascaded channel, we exploit angular sparsity in the channel matrices $\bm H_\mathrm{MS}$ and $\bm H_\mathrm{BS}$. For this, we apply the BEM by sampling the angular domain using a set of $N$ grid angles $\{\psi_n\}_{n=1}^N$ such that $\cos(\psi_n)= 2 n/N-1$ \cite{mao2022channel}. Then, \eqref{eq.channelmodel1} and~\eqref{eq.channelmodel2} reduce to
\begin{equation}
\label{eq.channelmodelsparse2}
\bm H_\mathrm{BS}= \bm A_{R} \bm g_{\mathrm{R}} \bm g_{\mathrm{L,d}}^\mathsf{H} \bm A_{L}^\mathsf{H}  \hspace{0.3cm} \text{and }  \hspace{0.3cm}
\bm H_\mathrm{MS}= \bm A_{L}\bm g_{\mathrm{L,a}} \bm g_{\mathrm{T}}^\mathsf{H} \bm A_{T}^\mathsf{H},
\end{equation}
where for any integer $Q>0$, using \eqref{eq.steer}, we define
\begin{equation}
\bm A_{Q} = \begin{bmatrix}
\bm a_{Q}(\psi_1) & \bm a_{Q}(\psi_2)&\ldots&\bm a_{Q}(\psi_N)
\end{bmatrix}\in \mathbb{C}^{Q\times N}.
\end{equation}
Also, $\bm g_{\mathrm{R}},\bm g_{\mathrm{L,d}},\bm g_{\mathrm{L,a}},\bm g_{\mathrm{T}}\in\mathbb{C}^{N\times 1}$ are the unknown sparse channel representations corresponding to the AoAs/AoD of the channel. 
Substituting \eqref{eq.channelmodelsparse2} into \eqref{eq.basicdatamodel}, and vectorizing the received signal $\{\bm Y_k\}_{k=1}^{K_I}$ followed by some algebraic simplifications (see \cite{he2022structure} for details), we arrive at 
\begin{equation}\label{eq.datamodel}
\tilde{\bm y} = (\tilde{\bm \Phi}_{\mathrm{L}} \otimes \bm \Phi_{\mathrm{T}} \otimes \bm \Phi_{\mathrm{R}}) (\bm g_{\mathrm{L,a}} \otimes \bm g_{\mathrm{L,d}}^* \otimes \bm g_{\mathrm{T}}^*\otimes\bm g_{\mathrm{R}}) + \tilde{\bm w}\in\mathbb{C}^{RK\times 1},
\end{equation}
where $\tilde{\bm \Phi}_{\mathrm{L}} = \bm \Theta^\mathsf{T}(\bm A_{L}^\mathsf{T} \odot \bm A_{L}^\mathsf{H})^\mathsf{T}$, $\bm \Phi_{\mathrm{T}}=\bm X^\mathsf{T} \bm A_{T}^*$, and $\bm \Phi_{\mathrm{R}}=\bm A_{R}$. Here, only the first $N$ columns of $\tilde{\bm \Phi}_{\mathrm{L}}$ are distinct~\cite{wang2020compressed}. Hence, removing the redundant columns to reduce the dimension of the representation leads to
\begin{equation}\label{eq.datamodelcs}
\tilde{\bm y} = (\bm \Phi_{\mathrm{L}} \otimes \bm \Phi_{\mathrm{T}} \otimes \bm \Phi_{\mathrm{R}})\bm g + \tilde{\bm w} = \tilde{\bm H}\bm g + \tilde{\bm w}\in\mathbb{C}^{RK\times 1},
\end{equation}
where $\bm \Phi_{\mathrm{L}}\in\mathbb{C}^{K_{\mathrm{I}}\times N}$ is formed by the first $N$ columns of $\tilde{\bm \Phi}_{\mathrm{L}}$ and $\tilde{\bm H}=\bm \Phi_{\mathrm{L}} \otimes \bm \Phi_{\mathrm{T}} \otimes \bm \Phi_{\mathrm{R}}\in\mathbb{C}^{RK\times N^3}$. Also, we define $\bm g = \bm g_{\mathrm{L}} \otimes \bm g_{\mathrm{T}}^* \otimes \bm g_{\mathrm{R}}\in\mathbb{C}^{N^3\times 1}$ with $\bm g_{\mathrm{L}}\in\mathbb{C}^{N\times 1}$ being the scaled version of the first $N$ entries of $\bm g_{\mathrm{L,a}} \otimes \bm g_{\mathrm{L,d}}^*$. Hence,~\eqref{eq.datamodelcs} translates the channel estimation problem into the sparse recovery problem in \eqref{eq.problem_basic} with unknown vector $\bm g$. Now we can apply our algorithms (AM-KroSBL and SVD-KroSBL) with $I=3$ to reconstruct the Kronecker-structured sparse channel vector $\bm g$. 

\subsection{Extension of AM-KroSBL for block sparsity}
In the IRS-aided channel model, the scatters lead to spreading AoAs (see Fig.~\ref{fig:irs_sys}), causing clustered non-zero BEM coefficients. In our model, sparse vectors $\bm g_{\mathrm{L}}$ and $\bm g_{\mathrm{R}}$, containing the BEM coefficients of the AoAs of IRS and BS, possess block sparsity structures with unknown block boundaries.
\begin{figure}[t!]
    \centering
    \begin{tikzpicture}[every node/.append style={scale=1}]
    
      \path (0,0) coordinate (origin);
      \node[IRS, element fill color=green!30, fill=red!30, minimum size=2.5cm] at (2.5,-3) (irs) {};
      \draw[] (3,-2.7) node[scale=1,rotate=0] {IRS};
      \draw[] (4.04,-2.1) node[scale=1,rotate=0] {$\ldots$};
    
      \path (-0.5,1) coordinate (mimoindposl);
      \path (7,1) coordinate (mimoindposr);
    
      \tikzset{test mimo node/.style={fill=blue!30,draw,thick,
          antenna offset=0.3cm, 
          minimum width=0.5cm,
          minimum height=2cm, 
          antenna base height=0.2cm,
          antenna side=0.4cm,
          antenna yshift=0.3cm
        }}
      \def\txrxsep{4cm}

      \draw (mimoindposl)  node[draw,test mimo node,mimoind] (mimoindtxnode) {BS};
      \draw (mimoindposr)  node[draw,test mimo node,mimoind,left antennas] (mimoindrxnode) {MS};
    
      \node[] (1) at (0.35,1.5)  {};
      \node[] (1B) at (0.35,1)  {};
      \node[] (12) at (1.8,-0.6)  {};  
      \node[] (121) at (0.35,0.6)  {};  
      \node[] (122) at (0.35,1.3)  {};  
      \node[] (123) at (0.35,2.0)  {};  
      
      \node[] (2) at (2.5,-1.8)  {};
      \node[] (2B) at (2,-1.8)  {};
      \node[] (2C) at (4,-1.8)  {};
      
      \node[] (23) at (4.9,-0.4)  {};
      \node[] (231) at (3.5,-1.8)  {};
      \node[] (232) at (3,-1.8)  {};
      \node[] (233) at (2.5,-1.8)  {};  
      
      \node[] (3) at (6.15,2)  {};
      \node[] (3B) at (6.15,1)  {};
      
      \draw[->] (23)--(231);
      \draw[->] (23)--(232);
      \draw[->] (23)--(233);
      
      \draw[->] (3)--(23);
      
      \draw[->] (2)--(12);
      \draw[->] (12)--(121);
      \draw[->] (12)--(122);
      \draw[->] (12)--(123);
      
      \draw[thick,black,->] (0.35,2.0) -- (0.35,0.5) node[black,right] {};
      \draw[thick,black,<->] (1.2,-1.8) -- (5,-1.8) node[black,right] {};
      \draw[thick,black,->] (6.13,2) -- (6.13,0.5) node[black,right] {};
      
      \draw pic[draw,angle radius=0.8cm,"$\alpha_{\mathrm{BS}}$",pic text options={shift={(0.5,0)}}] {angle=1B--123--12};
      \draw pic[draw,angle radius=0.8cm,"$\phi_{\mathrm{BS}}$",pic text options={shift={(-0.05,0)}}] {angle=12--2--2B};
      \draw pic[draw,angle radius=1.0cm,"$\phi_{\mathrm{MS}}$",pic text options={shift={(0.08,0)}}] {angle=2C--231--23};
      \draw pic[draw,angle radius=1cm,"$\alpha_{\mathrm{MS}}$",pic text options={shift={(-0.55,0)}}] {angle=23--3--3B};
      
      \node[cloud,
        draw =black,
        text=cyan,
        fill = gray!10,
        minimum width = 0.5cm,
        minimum height = 0.25cm] (c) at (4.9,-0.2) {Scatter 1};
        
      \node[cloud,
        draw =black,
        text=cyan,
        fill = gray!10,
        minimum width = 0.5cm,
        minimum height = 0.25cm] (c) at (1.8,-0.2) {Scatter 2};
    
    \end{tikzpicture}
    \caption{An illustration of AoAs and AoDs in an IRS-aided channel.}
    \label{fig:irs_sys}
\end{figure}
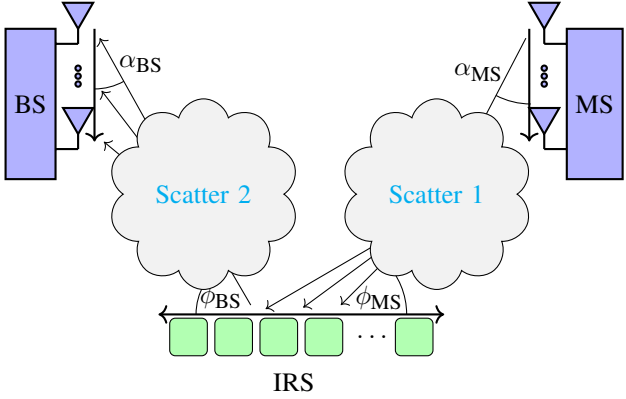

To tackle block sparsity, we draw inspiration from the PC-SBL algorithm~\cite{fang2014pattern} and impose a prior on each entry of the sparse vector, which not only depends on its hyperparameter but also the hyperparameters of its neighbors. This method connects the sparsity patterns of the adjacent entries, promoting block sparsity. We assume that $\bm \gamma_1$ exhibits block sparsity for ease of exposition. However, this idea can be readily extended when multiple hyperparameter vectors possess block sparsity. We adopt the prior on $\bm x$ with hyperparameters $\{\bm\gamma_i\}$ as
    $p(\bm x|\{\bm \gamma_i\}) = \mathcal{CN}\big(0,\hat{\bm \gamma}\big)$,
where $\hat{\bm \gamma} = \hat{\bm \gamma_1} \otimes (\otimes_{i=2}^I \bm \gamma_i)$ and 
\begin{equation}\label{eq.hatnohat}
    \hat{\bm \gamma_1} = \bm C_{\beta} \bm \gamma_1,
\end{equation}
where 
$\bm C_{\beta}\in \mathbb{R}^{N \times N}$ is a tridiagonal Toeplitz matrix with ones along its diagonal and $\beta$ along its first sub- and super-diagonals~\cite{noschese2013tridiagonal}. The parameter $\beta>0$ is the pattern-coupled coefficient. Using the new prior, the mean and variance of conditional distribution in the $r$th EM iteration are modified by replacing $\bm \gamma^{(r)}$ with $\hat{\bm \gamma}^{(r)}$ in~\eqref{eq.post_meva}.
Thus, the optimization problem in the M-step is
\begin{equation}\label{eq.mstep2}
\underset{\substack{\{\bm \gamma_i\}}}{\min}\log|\diag(\hat{\bm\gamma})|+(\bm d^{(r)})^{\mathsf{T}}\hat{\bm\gamma}^{-1} \ \text{s.t.} \ \left\{\hat{\bm \gamma_1}, \{\bm \gamma_i\}_{i=2}^I\right\} \in\mathcal{C}_+
.
\end{equation}
We solve~\eqref{eq.mstep2} using an iterative algorithm similar in spirit to the AM-KroSBL algorithm, by setting the gradient of  the cost function with respect to all hyperparameter vectors to zero. Here, the update for $\{\bm \gamma_i\}_{i=2}^I$ is given by~\eqref{eq.update}. However, due to the entanglement of hyperparameters in $\hat{\bm \gamma_1}$, the update for $\bm \gamma_1^{(r,t+1)}$ is
\begin{multline}\label{eq.gradg1compact}
     N^{I-1} \bm C_{\beta} (\hat{\bm \gamma_1}^{(r,t+1)})^{-1}
     =\bm C_{\beta} \diag\left((\hat{\bm \gamma_1}^{(r,t+1)})^{-2}\right)\\\times\left(\bm I_N\otimes (\otimes_{i=2}^I(\bm \gamma_i^{(r,t)})^{-1})\right)^\mathsf{T}\bm d^{(r)}.
\end{multline}
Solving~\eqref{eq.gradg1compact} is not trivial due to the matrix-vector multiplication on both sides. However, if $\bm C_\beta$ is invertible, we can simplify~\eqref{eq.gradg1compact} using \eqref{eq.hatnohat} as 
\begin{equation}
    \bm C_{\beta} \bm \gamma_1^{(r,t+1)}= \hat{\bm \gamma_1}^{(r,t+1)}=N^{-I+1}(\bm I_N\otimes(\otimes_{i=2}^I(\bm \gamma_i^{(r,t)})^{-1}))^\mathsf{T}\bm d^{(r)}.
\end{equation}
Also, we notice that the $n$-th eigenvalue of $\bm C_{\beta}$ is $\lambda_n = 1 + 2\beta \cos(\frac{n\pi}{N+1})$ for $n=1,2,\ldots,N$ \cite{noschese2013tridiagonal}. Thus, we choose $\beta \neq -(2\cos(\frac{n\pi}{N+1}))^{-1}$ so that $\bm C_{\beta}$ is non-singular.
With this choice, we update $\bm \gamma_1$ along with the constraint $\bm \gamma_1^{(r,t+1)} > 0$ with
\begin{equation}\label{eq.nnls1}
    \min_{\bm \gamma_1^{(r,t+1)} \geq 0} \|\bm C_{\beta} \bm \gamma_1^{(r,t+1)}
     \!-\!N^{-I+1}\left(\bm I_N\otimes (\otimes_{i=2}^I(\bm \gamma_i^{(r,t)})^{-1})\right)^\mathsf{T}\!\bm d^{(r)}\|^2_2.
\end{equation}
The resulting algorithm, namley PC-KroSBL, is summarized in Algorithm \ref{al.PCKroSBL}.

\begin{algorithm}[t]
\SetAlgoLined
\DontPrintSemicolon
\SetNoFillComment
  \KwData{Measurements $\bm y$, matrix $\bm H$, noise power $\sigma^2$}
  \textbf{Parameters}: Threshold $\epsilon$ and $\epsilon_{\mathsf{AM}}$
   
  \textbf{Initialization}: $\{\bm \gamma_i\}^{(0)}=\bm 0$, $\{\bm \gamma_i\}^{(1)} \in \mathcal{C}$, set $r=1$.
  
   \While{$\| \hat{\bm \gamma}^{(r)}-\hat{\bm \gamma}^{(r-1)} \|_2 > \epsilon$}{

    Compute $\bm d^{(r)}$ using~\eqref{eq.compute_d} and~\eqref{eq.post_meva} with $\bm\gamma^{(r)}=\hat{\bm \gamma}^{(r)}$ 

    Set $t=1$, $\{\bm \gamma_i\}^{(r,1)} = \{\bm \gamma_i\}^{(r)} \in \mathcal{C}_+$, $\{\bm \gamma_i\}^{(r,0)} = \bm 0$
    
    \While{$\| \hat{\bm \gamma}^{(r,t)}-\hat{\bm \gamma}^{(r,t-1)} \|_2 > \epsilon_\mathsf{AM}$}{ 

    Compute $\bm \gamma_1^{(r,t+1)}$ using~\eqref{eq.nnls1}
    
    Compute $\tilde{\bm \gamma_1}=\hat{\bm \gamma_1}^{(r,t+1)}$ using~\eqref{eq.hatnohat}
    
    Compute $\{\bm \gamma_i\}^{(r,t+1)}$ using \eqref{eq.update} and \eqref{eq.projection}
    

    Update AM iteration index $t \leftarrow t + 1$

  }

    $\{\bm \gamma_i\}^{(r+1)} = \{\bm \gamma_i\}^{(r,t)}$
  
    Update iteration index $r \leftarrow r + 1$

   }
   
   \KwResult{Output $\bm x=\bm \mu_{\bm x}$ using~\eqref{eq.post_meva}}
  \caption{PC-KroSBL}
  \label{al.PCKroSBL}
\end{algorithm}

\section{Theoretical Analysis of KroSBL}\label{sec.convergence}

This section focuses on the theoretical analysis of KroSBL. We discuss the convergence guarantee for AM-KroSBL in Sec.~\ref{sec.con_AM}. Then, we present our results on the values to which the algorithm converges by studying the local minima of the KroSBL cost function in~\eqref{eq.cost_ml}.

\subsection{Convergence property of AM-KroSBL}\label{sec.con_AM}

The convergence of AM-KroSBL is established using the properties of the EM algorithm, which is well studied in \cite{wu1983convergence}. It is known that under certain conditions, the EM algorithm guarantees convergence to stationary points of $\mathcal{L}$. Nonetheless, the guarantees of the EM algorithm in KroSBL depend on the convergence of the AM algorithm (inner loop). So, this section answers two questions: \emph{What are the convergence properties of the AM algorithm? Do the properties of AM guarantee the convergence of AM-KroSBL?} The first question is answered by Proposition \ref{thm.am_stationary}, serving as a cornerstone to the answer to the second question via Theorem \ref{thm.stationary}.
We first introduce a lemma that supports the main results.
\begin{lemma}\label{thm.am_cost}
 Consider the AM algorithm that solves the M-step optimization problem of the $r$th EM iteration of AM-KroSBL (Algorithm \ref{al.AMKroSBL}) given by \eqref{eq.mstep}, for a fixed iteration index $r>0$. If $\bm d^{(r)} > 0$ in \eqref{eq.compute_d}, then the cost function sequence $Q(\{\bm\gamma_i\}^{(r,t)})|_{t=1}^\infty$ generated by the AM algorithm is non-increasing.
\end{lemma}
\begin{proof}
    The non-increasing nature of sequence $Q(\{\bm\gamma_i\}^{(r,t)})|_{t=1}^\infty$ is because the AM algorithm in every iteration optimizes one hyperparameter vector while keeping the others fixed, i.e.,
\begin{multline}\label{eq.nonincrease}
    Q(\{\bm \gamma_i\}^{(r,t)})\geq Q\left(\tilde{\bm \gamma}_1,\{\bm \gamma_i^{(r,t)}\}_{i=2}^I\right)\geq Q\left(\{\tilde{\bm \gamma}_i\}_{i=1}^2\!,\!\{\bm \gamma_i^{(r,t)}\}_{i=3}^I\right) \\
    \geq Q(\{\tilde{\bm \gamma}_i\})=Q(\{\bm \gamma_i\}^{(r,t+1)}),
\end{multline}
where the last step follows because the projection step in \eqref{eq.projection} does not change the cost function value.
\end{proof}
The following proposition uses the above lemma to show that the iterates of the AM algorithm converges.
\begin{proposition}\label{thm.am_stationary}
[AM algorithm convergence] Consider the AM algorithm that solves the M-step optimization problem of the $r$th EM iteration of AM-KroSBL (Algorithm \ref{al.AMKroSBL}) given by \eqref{eq.mstep}, for a fixed iteration index $r>0$. If $\bm d^{(r)} > 0$ in \eqref{eq.compute_d}, then the sequence $\{\bm\gamma_i\}^{(r,t)}|_{t=1}^\infty$ converges to the set of stationary points of the M-step cost function $Q(\{\bm \gamma_i\})$ in $\mathcal{C}_+$.
\end{proposition}

\begin{proof}
See Appendix \ref{appe.am}.
\end{proof}

We note that Proposition \ref{thm.am_stationary} only guarantees the AM algorithm converges to a stationary point which is not necessarily a global minimum. However, the following result establishes that the convergence of AM to a stationary point is sufficient to ensure the convergence of AM-KroSBL.


\begin{thm}\label{thm.stationary}
Consider the model in~\eqref{eq.problem_basic} with the assumptions $i)$ the noise variance $\sigma^2 > 0$, $ii)$ there exists $\epsilon>0$  such that the dictionary satisfies $\Vert [\bm H]_i\Vert_2>\epsilon$, for $i=1,2,\ldots,N^I$, and $iii)$ the starting point of AM-KroSBL $\{\bm\gamma_i\}^{(1)} >0$. Then, the sequence $\{\bm\gamma_i\}^{(r)}|_{r=1}^{\infty}$ generated by AM-KroSBL (Algorithm \ref{al.AMKroSBL}) convergence to the set of the stationary points of its cost function $\mathcal{L}$ given by \eqref{eq.cost_ml}.
\end{thm}

\begin{proof}
See Appendix \ref{appe.em}.
\end{proof}
We note that the assumptions of Theorem \ref{thm.stationary} are realistic. In particular, the assumption on the dictionary holds when $\bm H$ has no zero columns. If the norm of a column in $\bm H$ is zero, indicating that all its elements are zero, then that column does not contribute to the measurement and can be removed. 

Furthermore, Proposition \ref{thm.am_stationary} suggests that $\{\bm\gamma_i\}^{(r)}\in\mathcal{C}_+$ and thus $\{\bm\gamma_i\}^{(r)}\!>\!0$, which seems to contradict the expected sparsity of the estimates $\{\bm\gamma_i\}^{(r)}$. However, $\{\bm\gamma_i\}^{(r)}\!>\!0$ only holds under the assumption $\bm d^{(r)}\!>\!0$ and the sequence $\{\bm d^{(r)}\}_{r=1}^\infty$ belongs to an open set $\{\bm d|\bm d > 0\}$. From our experiments, we observe that the sequence converges to $\bm d^{(\infty)}$ that belongs to the boundary of the open set, leading to sparse $\{\bm \gamma_i\}\in\mathcal{C}\setminus\mathcal{C}_+$. Intuitively, this behavior can be viewed as follows. If  the $n^*$th entry of $\bm \gamma_{i^*}^{(r)}$ goes to zero for some $i^*$ and $n^*$, then all the $N^{I-1}$ entries in $\bm\gamma^{(r)} = \otimes_{i=1}^I\bm\gamma_i^{(r)}$ involving  $[\bm \gamma_{i^*}]_{n^*}$ are zeros. Let $\mathcal{M}$ be the set of indices in $\otimes_{i=1}^I\bm\gamma_i^{(r)}$ approaching zero.  Then, the submatrix $[\bm \Sigma_{\bm x}]_{\mathcal{M}}$ goes to zero because from \eqref{eq.post_meva},
\begin{equation}
    \bm \Sigma_{\bm x} = \bm \Gamma^{(r)}-\bm \Gamma^{(r)}\bm H^{\mathsf{H}}\left(\sigma^2
    \bm I_{N^I}+\bm H\bm \Gamma^{(r)}\bm H^{\mathsf{H}}\right)\bm H \bm \Gamma^{(r)},
\end{equation}
where the rows of $ \bm \Gamma^{(r)}=\diag{\otimes_{i=1}^I\bm\gamma_i^{(r)}}$ indexed by $\mathcal{M}$ are close to zero. Consequently,  $[\bm d^{(r)}]_{\mathcal{M}}$ also goes to zero due to the following relation from \eqref{eq.compute_d} and \eqref{eq.post_meva},
\begin{equation}
    \bm d^{(r)} = \diag\left(\bm \Sigma_{\bm x}\left[\bm I_{N^I}+\sigma^{-4}\bm H^\mathsf{H}\bm y\bm y^\mathsf{H}\bm H \bm \Sigma_{\bm x}\right]\right). 
\end{equation}
Conversely, suppose that $[\bm d^{(r)}]_{\mathcal{M}}$  goes to zero, where $\mathcal{M}$ is index set of the entries in $\diag{\otimes_{i=1}^I\bm\gamma_i^{(r)}}$ corresponding to a particular element $[\bm \gamma_{i^*}^{(r)}]_{n^*}$ in $\bm \gamma_{i^*}^{(r)}$ for some $i^*$. Then, the second term in $Q$ of \eqref{eq.qfunc} involving $[\bm \gamma_{i^*}]_{n^*}$ vanishes and minimizing $\log [\bm \gamma_{i^*}]_{n^*}$ drives $[\bm \gamma_{i^*}^{(r)}]_{n^*}$ to zero. Thus, a sparse $\bm d^{(r)}$ encourages a sparse $\bm \gamma^{(r)}$ and vice versa, leading to a sparse convergent point $\{\bm \gamma_i\}^{(\infty)} \in \mathcal{C}\setminus\mathcal{C}_+$.

\subsection{Local minima of KroSBL cost function}
\label{sec:local_minima}
Having studied the convergence properties of the algorithm, we now look at the properties of the limit points. Unlike the previous section, the results of this section assume that the sparse vector is also Kronecker-structured. The first result of the section, Theorem~\ref{thm.local_minima_sparse}, proves that all local minima $\{\bm \gamma_i\}$ of the KroSBL cost function $\mathcal{L}$ in~\eqref{eq.cost_ml} are sparse. Subsequently, in Theorem~\ref{thm.no_local}, we derive an upper bound on the number of local minima of the KroSBL cost function.

\begin{thm}\label{thm.local_minima_sparse}
    In the noiseless setting, every local minimum of $\mathcal{L}$ is achieved at a sparse solution $\{\bm \gamma_i\}$, that is, $\|\bm \gamma_i\|_0 \leq {M_i}$ for $i=1,2,\cdots,I$, if the sparse vector is Kronecker-structured, i.e., $\bm x =\otimes_i^I\bm x_i$, for some $\bm x_i\in\mathbb{C}^{N}$.
\end{thm}
\begin{proof}
    See Appendix \ref{appe.sparse_local}.
\end{proof}


Theorem \ref{thm.local_minima_sparse} 
indicates that the local minimum is sparse, not because some hyperparameter vectors ($\bm\gamma_i$'s) are dense while others are sparse, leading to a sparse Kronecker product. Instead, it implies each $\bm\gamma_i$ generated by KroSBL is sparse, following our  Kronecker-structured support model. 

Now we discuss an upper bound for the number of local minima of $\mathcal{L}$. For this, we use the concept of unique representation property (URP). The matrix $\bm H$ is said to satisfy URP if any subset of ${\bar{M}}$ columns of $\bm H$ is linearly independent~\cite{wipf2004sparse}. If the dictionary $\bm H$ satisfies the URP, the number of local minima of $\mathcal{L}$ in \eqref{eq.cost_ml} without the Kronecker-structured support constraint $\mathcal{C}$ (the classic SBL algorithm) is~\cite{wipf2004sparse}
\begin{equation}
    \mathcal{N}_{\mathsf{SBL}} \leq \binom{N^I}{{\bar{M}}} - \sum_{p=1}^{P} \binom{N^I - D_{p}}{{\bar{M}}- D_{p}} + P\leq \binom{N^I}{{\bar{M}}},
\end{equation}
where $D_{p}$ is $\ell_0$-norm of the $p$th degenerate sparse solution of the SBL cost function, and $P$ is the number of sparse solution. When we impose the Kronecker-structured support constraint, the upper bound of the number of local minima decreases, as discussed~next.

\begin{thm}\label{thm.no_local}
Consider the model in~\eqref{eq.problem_basic} and assume that i) the noise variance $\sigma^2 = 0$, ii) $\bm H$ satisfies the URP, and iii) there exist $P_i$ degenerate sparse solutions $\bm x_{i,p}$, $p=1,2,\ldots,P_i$ such that $\bm y = (\otimes_{i=1}^I \bm H_i)(\otimes_{i=1}^I \bm x_{i,p})$ and 
    $\Vert\bm x_{i,p}\Vert_0 = D_{i,p} < {M_i}$. Then, the number of distinct local minima of the KroSBL cost $\mathcal{L}$ in $\mathcal{C}$, denoted as $\mathcal{N}$, satisfies
    \begin{equation}\label{eq.upper_local_minima}
        \mathcal{N} \leq \prod_{i=1}^I \left(  \binom{N}{{M_i}} -\sum_{p=1}^{P_i} \binom{N - D_{i,p}}{{M_i}- D_{i,p}} + P_i \right)\leq \prod_{i=1}^I\binom{N}{{M_i}}.
    \end{equation}
\end{thm}

\begin{proof}
    See Appendix \ref{appe.upper_bound}.
\end{proof}

%

Theorem \ref{thm.no_local} and the result for classical SBL uses the same assumption, $\bm H$ satisfies URP, to derive an upper bound for the distinct number of local minima. 
However, our result shows that $\mathcal{N}$ is dominated by $\prod_{i=1}^I\binom{N}{{M_i}}$ while $\mathcal{N}_{\mathsf{SBL}}$ is dominated by $\binom{N^I}{{\bar{M}}}>\prod_{i=1}^I\binom{N}{{M_i}}$. Thus, incorporating the Kronecker structure can greatly diminish the solution space, explaining the better reconstruction performance of the KroSBL.

\section{Performance Evaluation}
\label{sec:numsimu}

We conduct numerical experiments to investigate the efficacy of our algorithms for sparse vector recovery. We evaluate the recovery performance of AM-KroSBL and SVD-KroSBL by comparing them with three benchmarking algorithms: the classic SBL (cSBL) \cite{wipf2004sparse}, KSBL \cite{chang2021sparse}, and KOMP \cite{caiafa2013computing}. 

We choose $I=3$ i.e., $\bm H = \otimes_{i=1}^3\bm H_i$ and the sparse vector $\bm x = \otimes_{i=1}^3\bm x_i$ where $\bm x_i\in\mathbb{R}^{15\times 1}$. The entries of $\bm x_i\in\mathbb{R}^{15\times 1}$, $\bm H_1 \in \mathbb{R}^{M\times 15}$, $\bm H_2 \in \mathbb{R}^{12\times 15}$, and $\bm H_3 \in \mathbb{R}^{15\times 15}$ are drawn from $\mathcal{N}(0,1)$, where $M = \{6,8,10,12,14\}$, called measurement level, controls the total number of measurements $\bar{M}=180M$ (or the under-sampling ratio $\bar{M}/N^I=180M/3375)$. The sparsity level for each $\bm x_i$ is $S=\{2,3,4,5,6\}$, and the support is generated uniformly at random. The zero-mean additive white Gaussian measurement noise level is decided by the signal-to-noise ratio $\text{SNR~(dB)} = 10\log_{10}\mathbb{E}\{\|\bm H\bm x\|_2^2/\|\bm n\|_2^2\}$ and takes values from $\{5,10,15,20,25,30\}$.

We use three metrics for the assessment: NMSE, support recovery rate (SRR), and run time. Here, we define 
\begin{align}
    \text{NMSE} &= \mathbb{E}\left\{\frac{\|\bm x-\hat{\bm x}\|_2}{\|\bm x\|_2}\right\}\\
    \text{SRR} &= \frac{|\supp(\hat{\bm x}) \cap \supp(\bm x) |}{|\supp(\hat{\bm x}) \cup \supp(\bm x) |},
\end{align}
where ${\bm x}$ is the ground truth and $\hat{\bm x}$ is the reconstructed vector. 
We set the maximum EM iteration to 150 EM for the SBL-based methods. We also implement the complexity reduction technique described in \cite{he2022structure} for AM- and SVD-KroSBL and the technique in \cite{chang2021sparse} for KSBL. The pruning is also included to prune small hyperparameters for the all SBL-based methods \cite{zhang2011clarify}. We average over a hundred independent realizations. Our observations from the results, summarized in Figs.~\ref{fig_con} to \ref{fig_ce} and Tables~\ref{tab:sparse_recovery} and \ref{tab.time_pc}, are as follows\footnote{
Our code is available \href{https://github.com/YanbinHe/JournalKroSBL.git}{here}. Also, see \href{https://github.com/Dingqinliu/Encryption_Matlab/blob/3fd2edaadf10b512810933465361cd2ee1af1337/encryption_based_on_CS_chaotic/Tensor_CS/Fig_8/tensor_OMPND.m}{GitHub link} for the KOMP code.}.
\subsection{Convergence illustration}

\begin{figure}[t!]
    \centering
  \subcaptionbox{NMSE without pruning step\label{fig_con1}}{\includegraphics[width=0.3\textwidth]{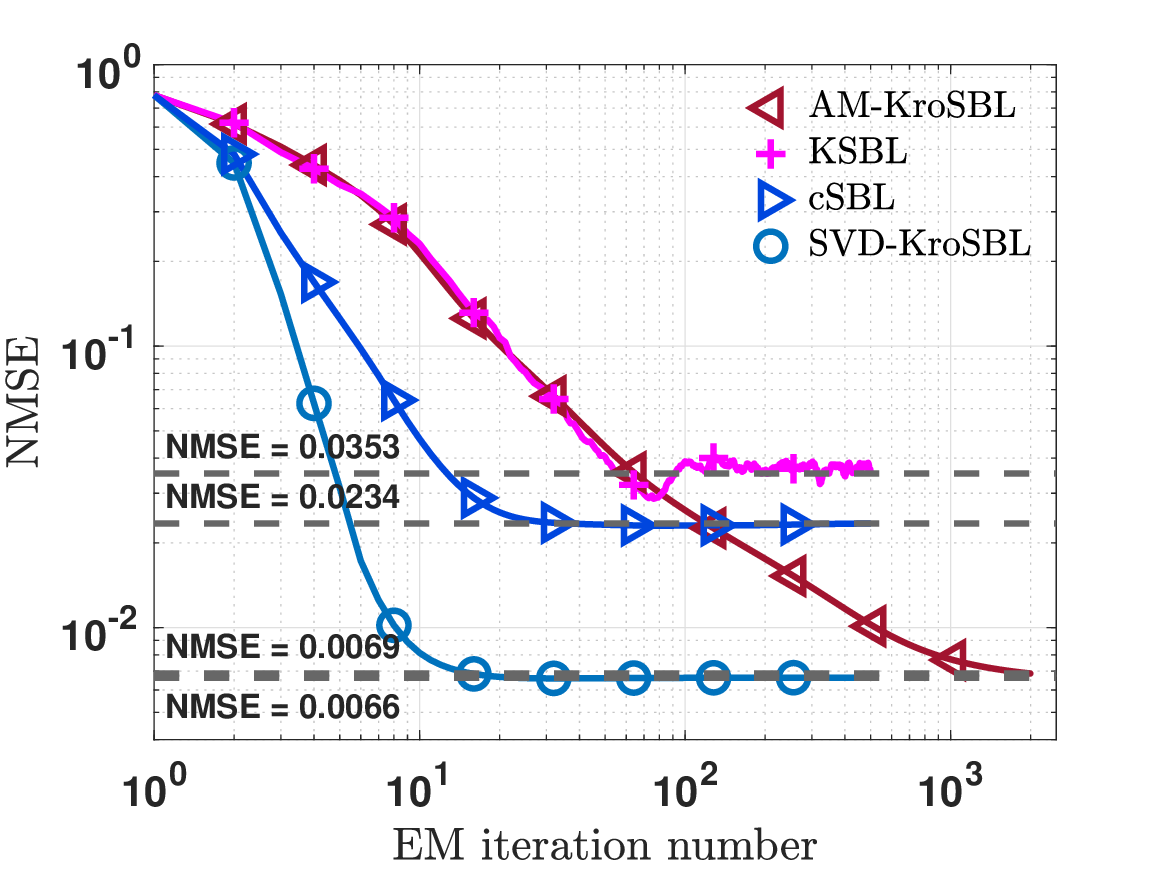}}
  \subcaptionbox{NMSE with pruning step\label{fig_con2}}
  {\includegraphics[width=0.3\textwidth]{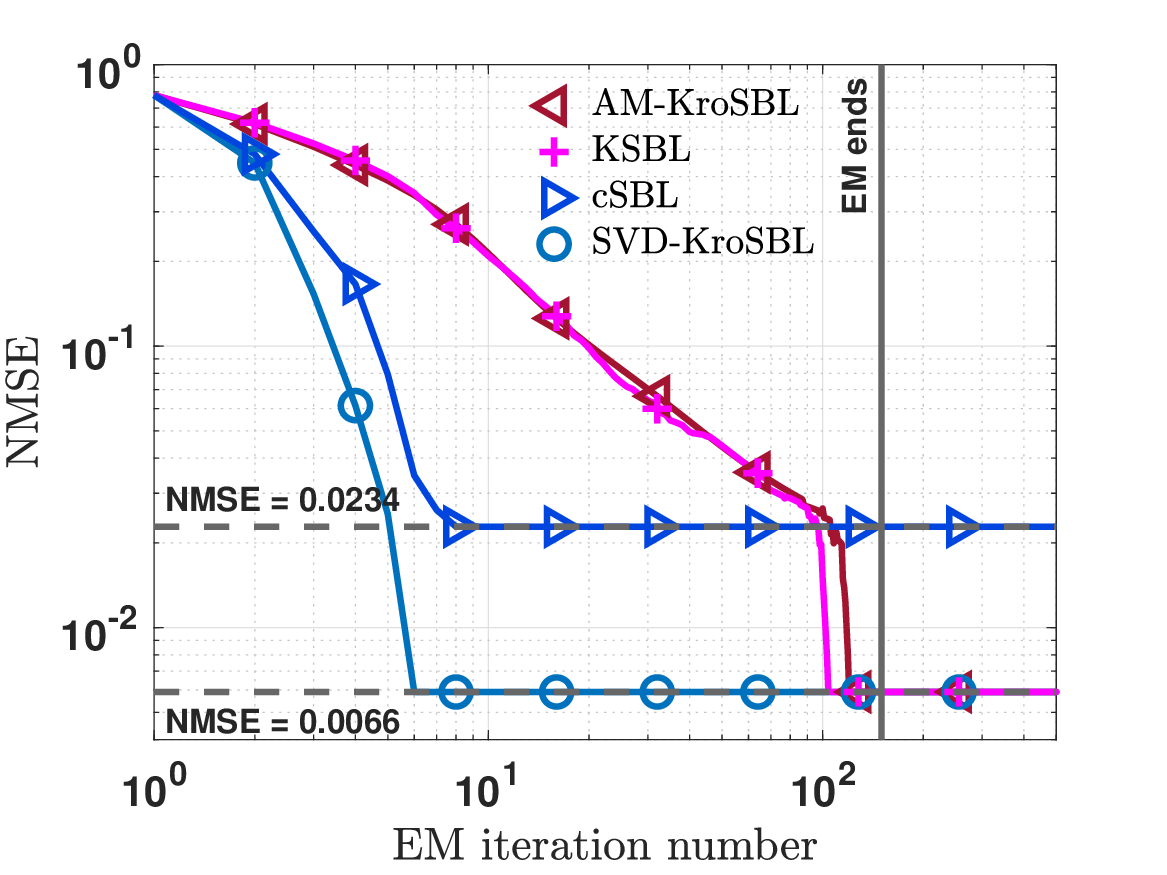}}
  \caption{Convergence plots of AM-, SVD-KroSBL, KSBL, and cSBL with measurement level $M = 14$, sparsity $S=4$, and $\text{SNR} = 30$dB}
    \label{fig_con}
\end{figure}

Fig. \ref{fig_con} demonstrates the convergence property of cSBL, KSBL, AM-, and SVD-KroSBL \emph{with and without pruning}. We include the convergence without pruning because our theoretical analysis does not account for pruning. First, we look at Fig. \ref{fig_con}. We note that our AM-KroSBL and SVD-KroSBL lead to lower NMSE because they incorporate the Kronecker structure, avoiding unwanted local minimum compared with cSBL. However, KSBL, despite incorporating the same prior knowledge, results in a higher NMSE than AM- and SVD-KroSBL. Interestingly, KSBL initially has a similar NMSE as AM-KroSBL, as both schemes solve~\eqref{eq.mstep} using the gradient-based iterative method. However, KSBL is trapped in a local minimum after a few iterations because it uses a loose approximation and does not constrain its hyperparameters $\{\bm \gamma_i\}$ in $\mathcal{C}$, i.e., does not normalize the hyperparameters. Our experiments show that the entries of some $\bm \gamma_i$'s of KSBL upon convergence become very small while others become very large. Although the Kronecker product of $\{\bm \gamma_i\}$ of both algorithms are the same initially, small $\bm \gamma_i$ values unstabilize KSBL. This numerical instability reduces the estimation accuracy and leads to a local minimum with high NMSE. 
To mitigate this issue, pruning small components is useful, helping KSBL to converge numerically, as shown in Fig.~\ref{fig_con2}. Pruning also accelerates other schemes. However, NMSE performance is sensitive to the pruning threshold, which is only empirically optimized. A larger threshold leads to faster convergence but is at the risk of eliminating true components.
\subsection{Comparison with the state-of-the-art}

\begin{figure*}[t!]
    \centering
  \subcaptionbox{$S=3$ and $M=8$\label{fig1.a}}{\includegraphics[width=0.33\textwidth]{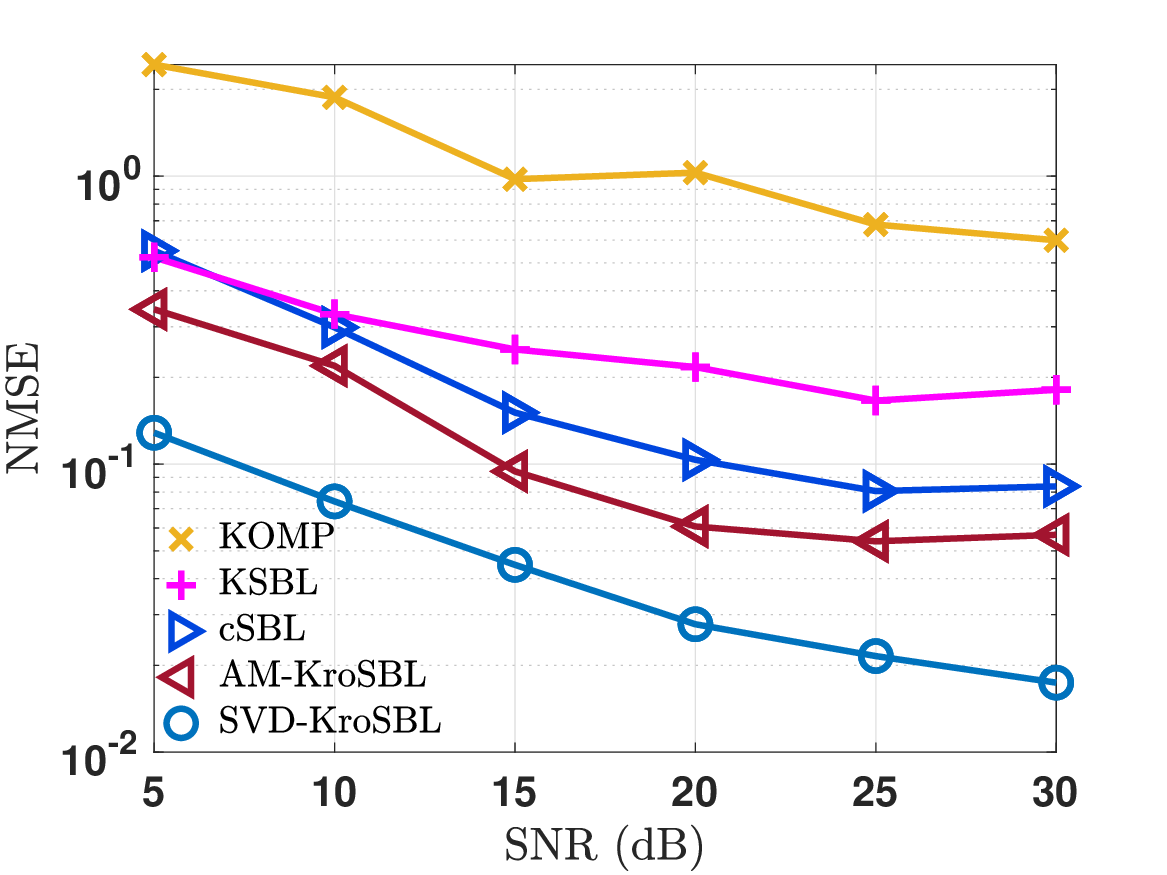}}\hspace{0em}%
  \subcaptionbox{$S=3$ and $\text{SNR} = 25\text{dB}$\label{fig1.b}}{\includegraphics[width=0.33\textwidth]{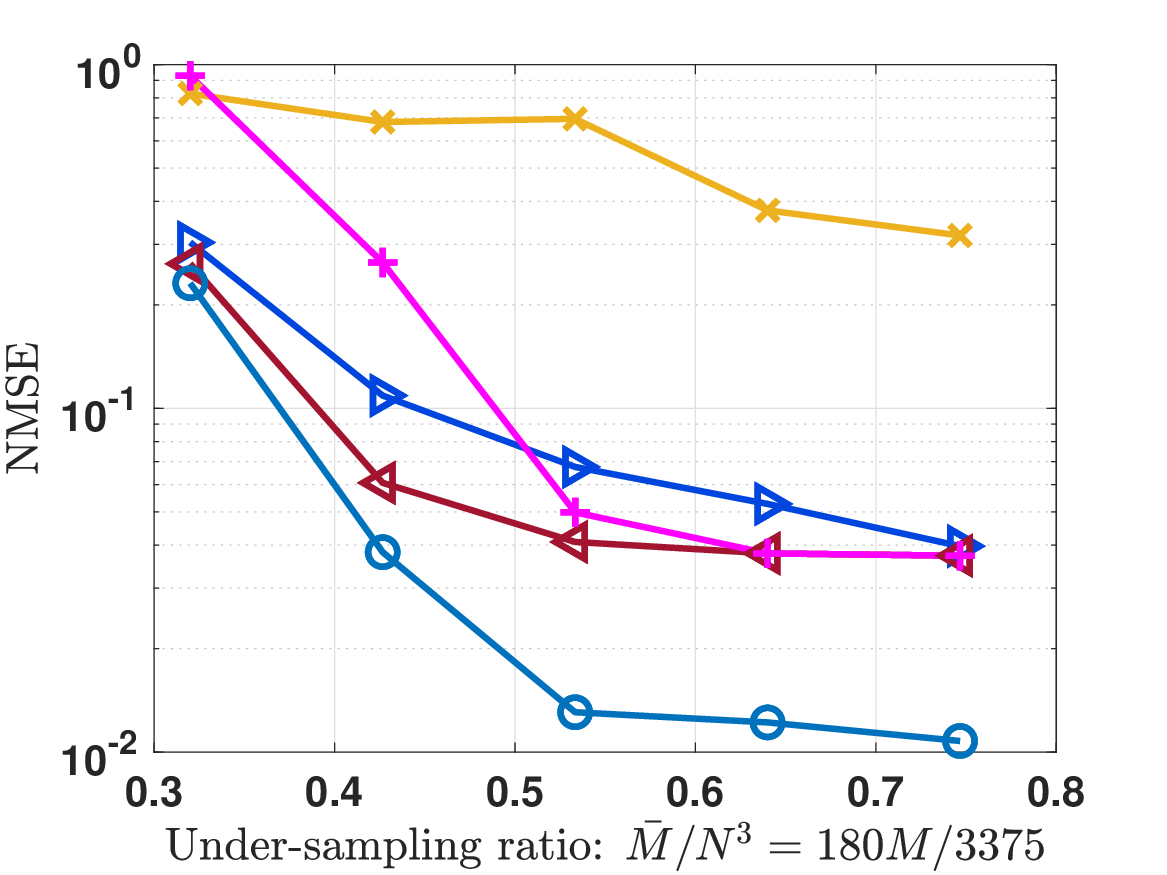}}\hspace{0em}
  \subcaptionbox{$M=8$ and $\text{SNR} = 25\text{dB}$\label{fig1.c}}{\includegraphics[width=0.33\textwidth]{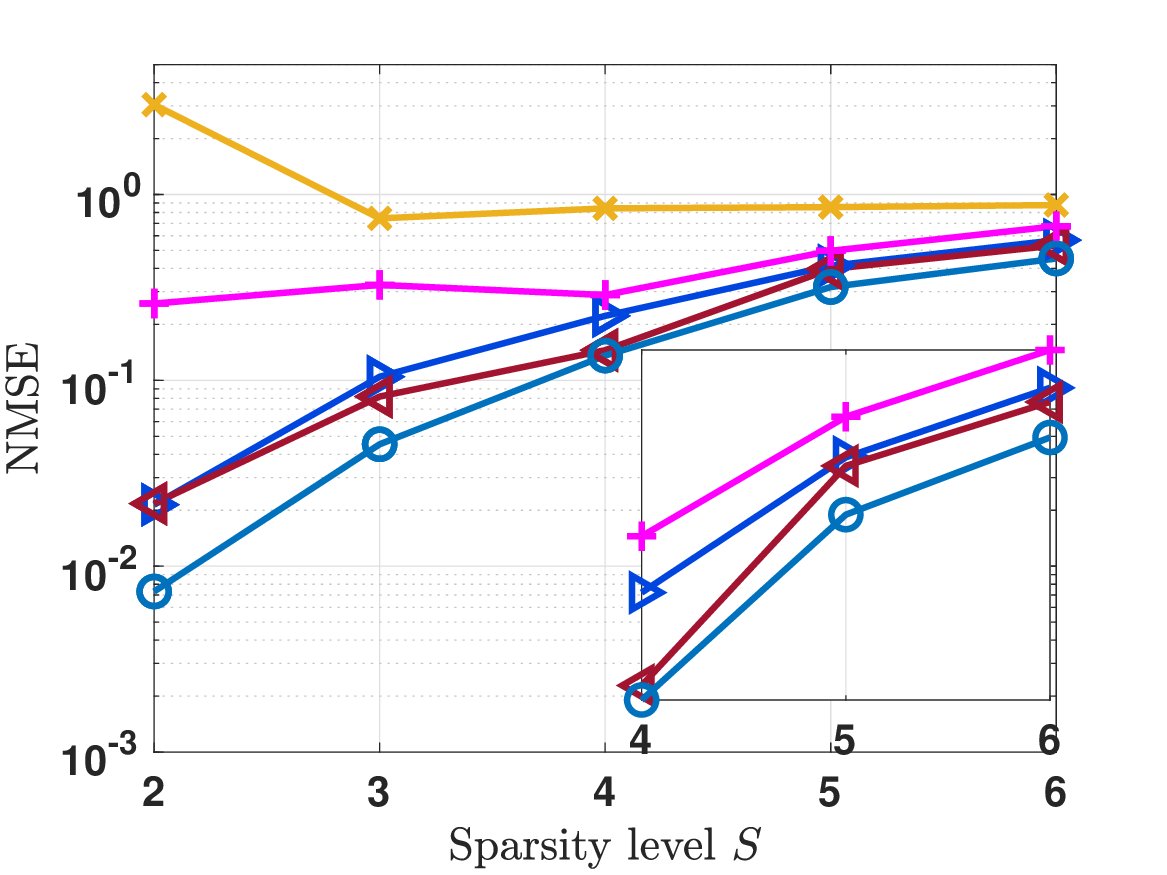}}
  \subcaptionbox{$S=3$ and $M = 8$\label{fig2.a}}{\includegraphics[width=0.33\textwidth]{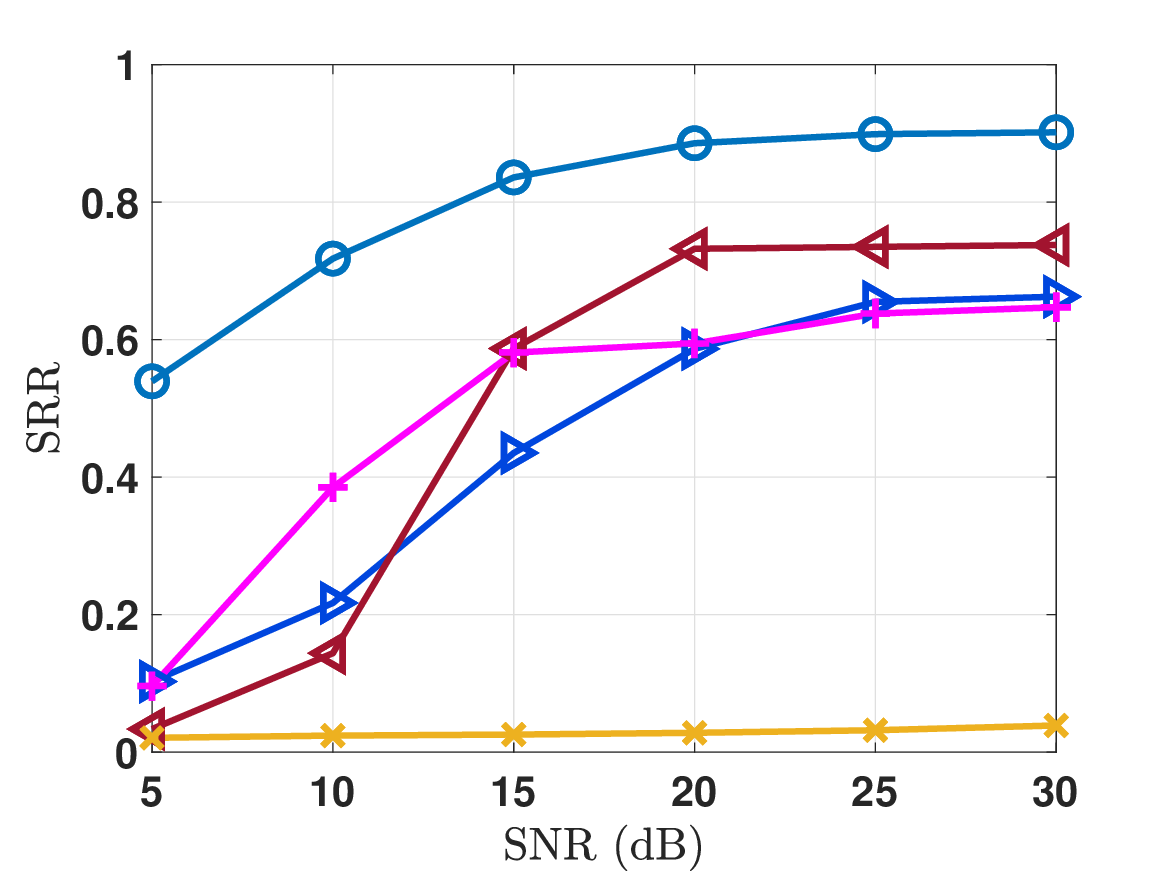}}\hspace{0em}%
  \subcaptionbox{$S=3$ and $\text{SNR}=25\text{dB}$\label{fig2.b}}{\includegraphics[width=0.33\textwidth]{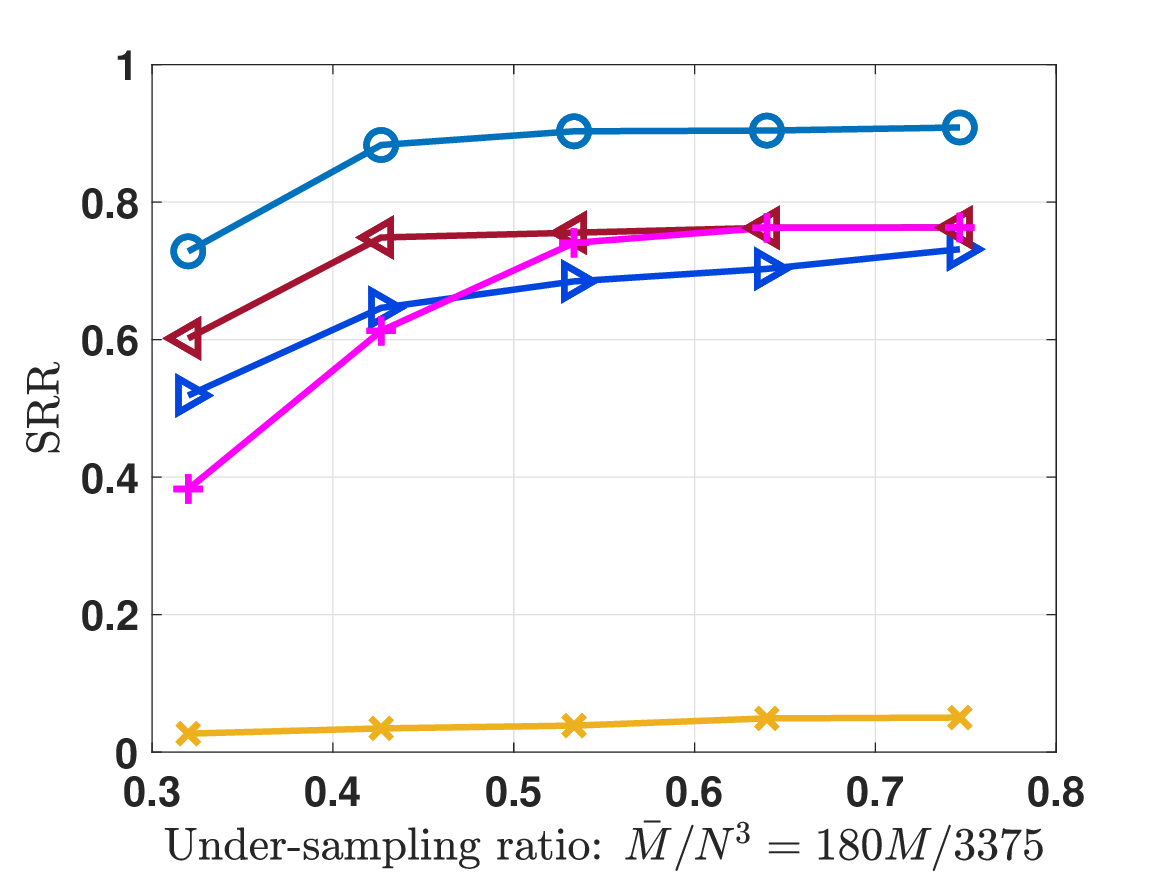}}\hspace{0em}
  \subcaptionbox{$M=8$ and $\text{SNR} =25\text{dB}$\label{fig2.c}}{\includegraphics[width=0.33\textwidth]{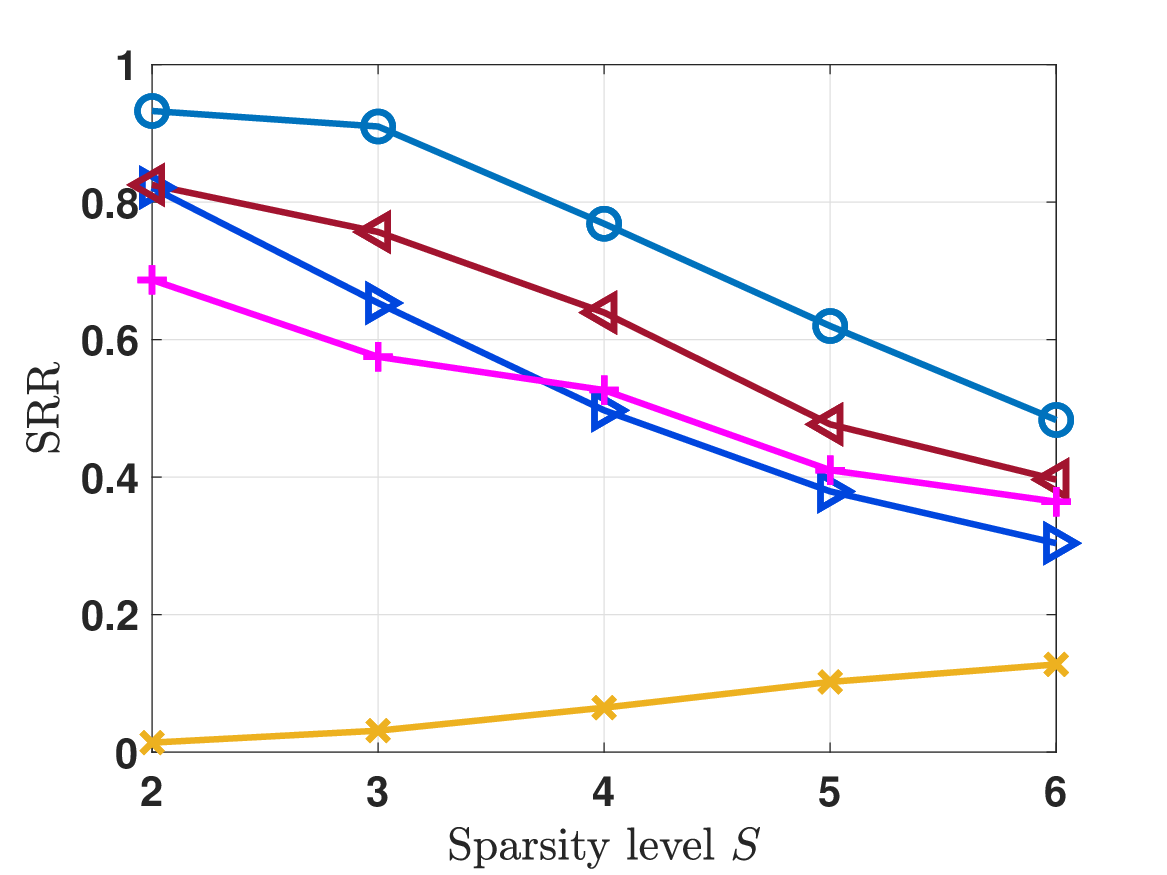}}
    \caption{NMSE and SRR performance of different algorithms as functions of SNR, under-sampling ratio, and sparsity level $S$. Unless otherwise mentioned in the plots, measurement level $M=8$, sparsity level $S=3$, and $\text{SNR} = 25~\text{dB}$.}
    \label{fig.snr}
\end{figure*}

Figs. \ref{fig1.a} and \ref{fig1.b} show that the performance of all the algorithms improves with SNR and the number of measurements $\bar{M}=180M$ (or equivalently under-sampling ratio $\bar{M}/N^I=180M/3375)$, as expected. Similarly, Fig. \ref{fig1.c} shows that increasing the sparsity level hinders the reconstruction of all schemes, as the number of measurements is unchanged. Further, our SVD-KroSBL outperforms all the other algorithms, both in terms of NMSE and SRR. The next best-performing algorithm is AM-KroSBL in most cases. The exception is in Fig. \ref{fig2.a}, where KSBL and cSBL have higher SRR than the AM-based when the SNR is low. This behavior is because the pruning step of the AM-based fails to eliminate small components outside the true support. However, the nonzero entries are recovered accurately, indicated by a low NMSE. Finally, KSBL is expected to perform better than cSBL. However, this is not true in the low-measurement regime due to its approximations, seen in  Fig. \ref{fig1.b}. 

Table \ref{tab:sparse_recovery} compares the run time of different algorithms. KOMP has the lowest run time due to its greedy nature but suffers from high NMSE and low SRR, as shown in Fig.~\ref{fig.snr}. SVD-KroSBL has the lowest run time among the remaining algorithms. It is faster as it is non-iterative and takes fewer EM iterations to converge (see Fig.~\ref{fig_con}), and also uses the complexity reduction technique \cite{he2022structure} in the E-step.
We also observe that AM-KroSBL outperforms the KSBL and has a similar run time as cSBL. Although cSBL takes fewer EM iterations, each EM iteration of AM-KroSBL is faster than cSBL due to the complexity reduction technique \cite{he2022structure}. 
Further, unlike KSBL and cSBL, AM-KroSBL has an iterative M-step, but the inner loop converges quickly, making it faster than KSBL. 

\begin{table}[]
\caption{Computation time (in seconds) comparison with $\text{SNR}=25\text{dB}$ and measurement level $M=8$ (underlined text for \underline{the best} and boldface for \textbf{the second best})} 
\scriptsize
\centering
\begin{tabular}{llllll}
\hline
\multicolumn{1}{l|}{Sparsity level}                     & \multicolumn{1}{l|}{$S=2$}                       & \multicolumn{1}{l|}{$S=3$}                       & \multicolumn{1}{l|}{$S=4$}                       & \multicolumn{1}{l|}{$S=5$}                       & $S=6$                       \\ \hline
\hline
\multicolumn{1}{l|}{\multirow{1}{*}{AM-KroSBL}} & \multicolumn{1}{l|}{\multirow{1}{*}{14.068}} & \multicolumn{1}{l|}{\multirow{1}{*}{12.863}} & \multicolumn{1}{l|}{\multirow{1}{*}{11.215}} & \multicolumn{1}{l|}{\multirow{1}{*}{10.180}} & \multirow{1}{*}{10.070} \\ \hline
\multicolumn{1}{l|}{SVD-KroSBL}                 & \multicolumn{1}{l|}{\textbf{3.520}}          & \multicolumn{1}{l|}{\textbf{2.945}}          & \multicolumn{1}{l|}{\textbf{2.596}}          & \multicolumn{1}{l|}{\textbf{2.651}}          & \textbf{3.008}          \\ \hline
\multicolumn{1}{l|}{KSBL}                     & \multicolumn{1}{l|}{23.700}                  & \multicolumn{1}{l|}{22.858}                  & \multicolumn{1}{l|}{20.537}                  & \multicolumn{1}{l|}{19.413}                  & 18.632                  \\ \hline
\multicolumn{1}{l|}{cSBL}                       & \multicolumn{1}{l|}{12.837}                  & \multicolumn{1}{l|}{11.449}                  & \multicolumn{1}{l|}{10.510}                  & \multicolumn{1}{l|}{11.541}                  & 12.723                  \\ \hline
\multicolumn{1}{l|}{KOMP}                       & \multicolumn{1}{l|}{{\ul \textit{1.170}}}    & \multicolumn{1}{l|}{{\ul \textit{1.257}}}    & \multicolumn{1}{l|}{{\ul \textit{1.170}}}    & \multicolumn{1}{l|}{{\ul \textit{1.179}}}    & {\ul \textit{1.250}}    \\ \hline

\end{tabular}
\label{tab:sparse_recovery}
\vspace{0.2cm}
\end{table}

\subsection{Comparison of SVD-KroSBL and AM-KroSBL}\label{sec.svd_best}

We observe from Fig.\ref{fig.snr} that the SVD-KroSBL algorithm outperforms AM-KroSBL and cSBL with sufficient measurements regardless of its approximations. Here, we give the intuition behind the better performance of the SVD-KroSBL. In KroSBL, the M-step solves~\eqref{eq.mstep}. SVD-KroSBL approximates~\eqref{eq.mstep} by first identifying~\eqref{eq.mstep_svd} and then solving~\eqref{prob.bidecom}. Suppose $\bm d^{(r)} = \otimes_{i=1}^I \bm d_i^{(r)}$ for some nonnegative vectors $\bm d_i^{(r)}\in\mathbb{R}^{N\times 1}$. Then the optimal solution to \eqref{eq.mstep} is attained at $\check{\bm\gamma}_i=\bm d_i^{(r)}/\Vert\bm d_i^{(r)}\Vert$ for $i=1,2,\ldots,I-1$, and $\check{\bm\gamma}_I=\prod_{i=1}^{I-1}\Vert\bm d_i^{(r)}\Vert\bm d_I^{(r)}$. Since \eqref{eq.mstep_svd} attains the optimal value at $\check{\bm\gamma}=\bm d^{(r)}$ that satisfies the constraints of \eqref{eq.mstep}, $\check{\bm\gamma}=\bm d^{(r)}$ is also optimal for \eqref{eq.mstep}. Hence, the SVD-based is exact when $\bm d^{(r)}$ is Kronecker-structured. Further, we empirically verified the rank of rearranged/devectorized $\bm d^{(r)}$ when $I = 3$, where $\text{rank} = 1$ indicates $\bm d^{(r)} = \otimes_{i=1}^3 \bm d_i^{(r)}$. Without noise, the SVD-based method drives $\bm d^{(r)}$ to this structure within a few EM iterations and retains this structure, as shown in Fig. \ref{fig.rank}. Therefore, after a few EM iterations, the SVD-based method solves the M-step exactly and outperforms the iterative first-order optimization in the AM-KroSBL algorithm. Furthermore, Fig. \ref{fig.rank} also indicates that SVD-based converges to a Kronecker-structured $\bm d^{(r)}$ faster than AM-KroSBL and cSBL. Thus, SVD can be viewed as a stronger imposement of the Kronecker structure that accelerates EM compared with the other algorithms. 

\begin{figure}[t!]
\centering
\includegraphics[width=0.4\textwidth]{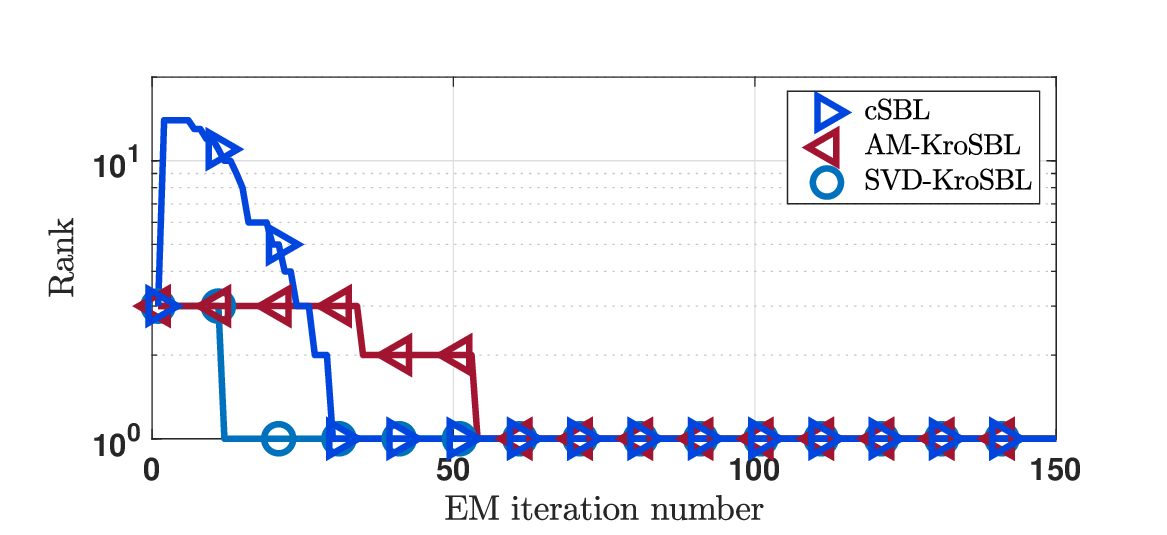}
\caption{The rank of reorganized $\bm d^{(r)}$ generated by different schemes in the noiseless setting with measurement level $M = 14$ and sparsity level $S=4$.}
\label{fig.rank}
\end{figure}


\subsection{IRS-aided channel estimation}
In this subsection, we present the observation from the results obtained when our PC-KroSBL is applied to the IRS-aided MIMO channel estimation, as discussed in Sec.~\ref{sec:channelesti}. We choose the number of BS antennas $R=16$, the number of MS antennas $T=6$, and the number of IRS elements $L=256$. Each entry of the IRS configurations $\{\bm \theta_k\}_{k=1}^{K_{\mathrm{I}}}$ is uniformly drawn from $\{-1/\sqrt{N}$ $,1/\sqrt{N}\}$ with $K_{\mathrm{I}}=6$ and $K_{\mathrm{P}}=3$. The number of grid angles is $N=18$, and all AoDs/AoAs are drawn from the grids. We assume the angle spreads over three grid points, leading to three consecutive non-zeros in sparse vectors $\bm g_{\mathrm{L}}$ and $\bm g_{\mathrm{R}}$. For this, we choose one grid point uniformly at random and add the selected grid and its two neighboring grids to the support. The channel gains $\beta_{{\mathrm{MS}},p}$ and $\beta_{{\mathrm{BS}},p}$ defined in Sec.~\ref{sec:channelesti} are drawn from $\mathcal{CN}(0,1)$~\cite{lin2021channel}. The performance of our PC-KroSBL is compared with PC-SBL \cite{fang2014pattern} and cSBL \cite{wipf2004sparse}. The pattern-coupled coefficient is $0.9$ in the PC-SBL algorithm. We choose $\beta=0.5$ and $\beta=0.8$. The performance metrics are NMSE, symbol error rate (SER), and run time. Here, channel estimation NMSE is given as
\begin{equation*}
\frac{1}{K_{\mathrm{I}}}\sum_{k = 1}^{K_{\mathrm{I}}}\!\!\frac{\|\bm H_\mathrm{BS} \diag (\bm \theta_{k}) \bm H_\mathrm{MS} - \tilde{\bm H}_\mathrm{BS} \diag (\bm \theta_{k}) \tilde{\bm H}_\mathrm{MS}\|_F^2}{\|\bm H_\mathrm{BS} \diag (\bm \theta_{k}) \bm H_\mathrm{MS}\|_F^2},
\end{equation*}
with $\tilde{\bm H}_\mathrm{BS} \diag (\bm \theta_{k}) \tilde{\bm H}_\mathrm{MS}$ being the reconstructed channel. SER is computed over data transmission containing $10^5$ uncoded QPSK symbols. The results are averaged over fifty independent channel realizations.

  \begin{figure}[t!]
    \centering
  \includegraphics[width=0.3\textwidth]{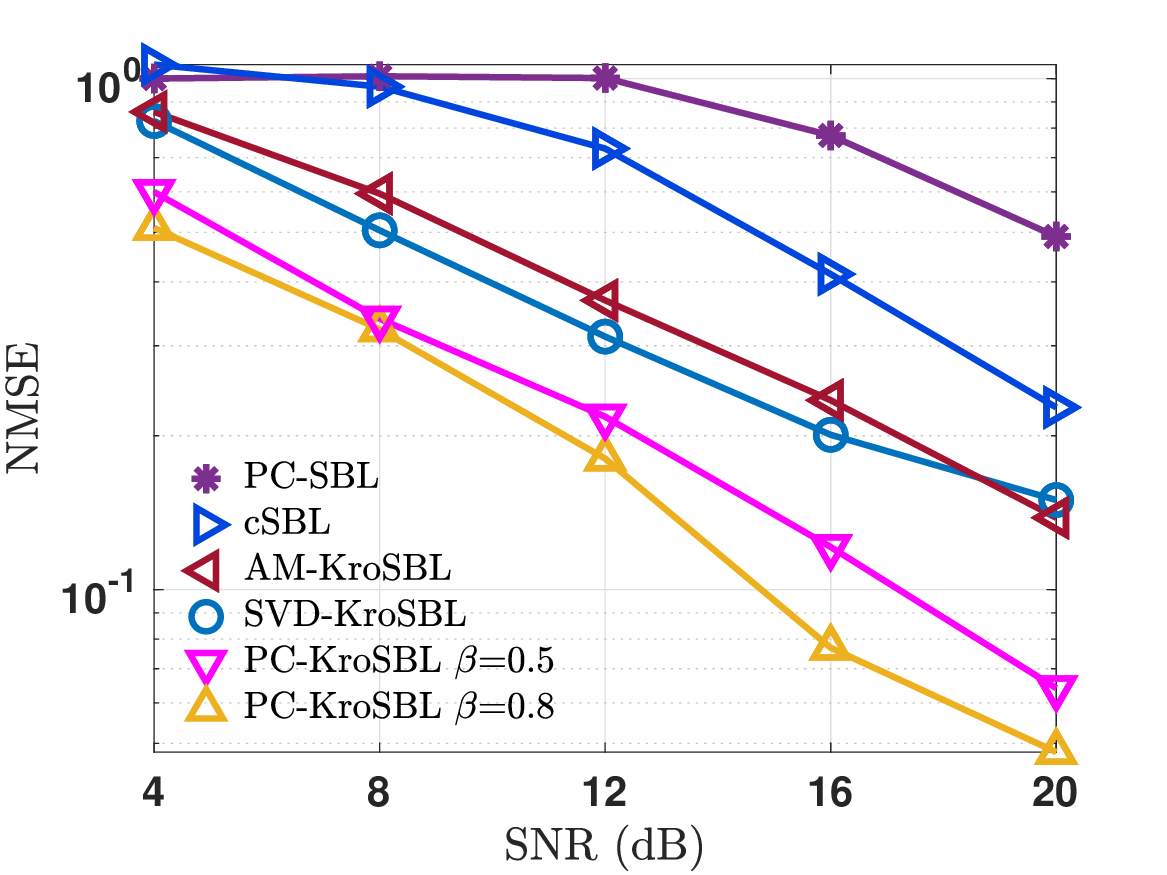}
  \includegraphics[width=0.3\textwidth]{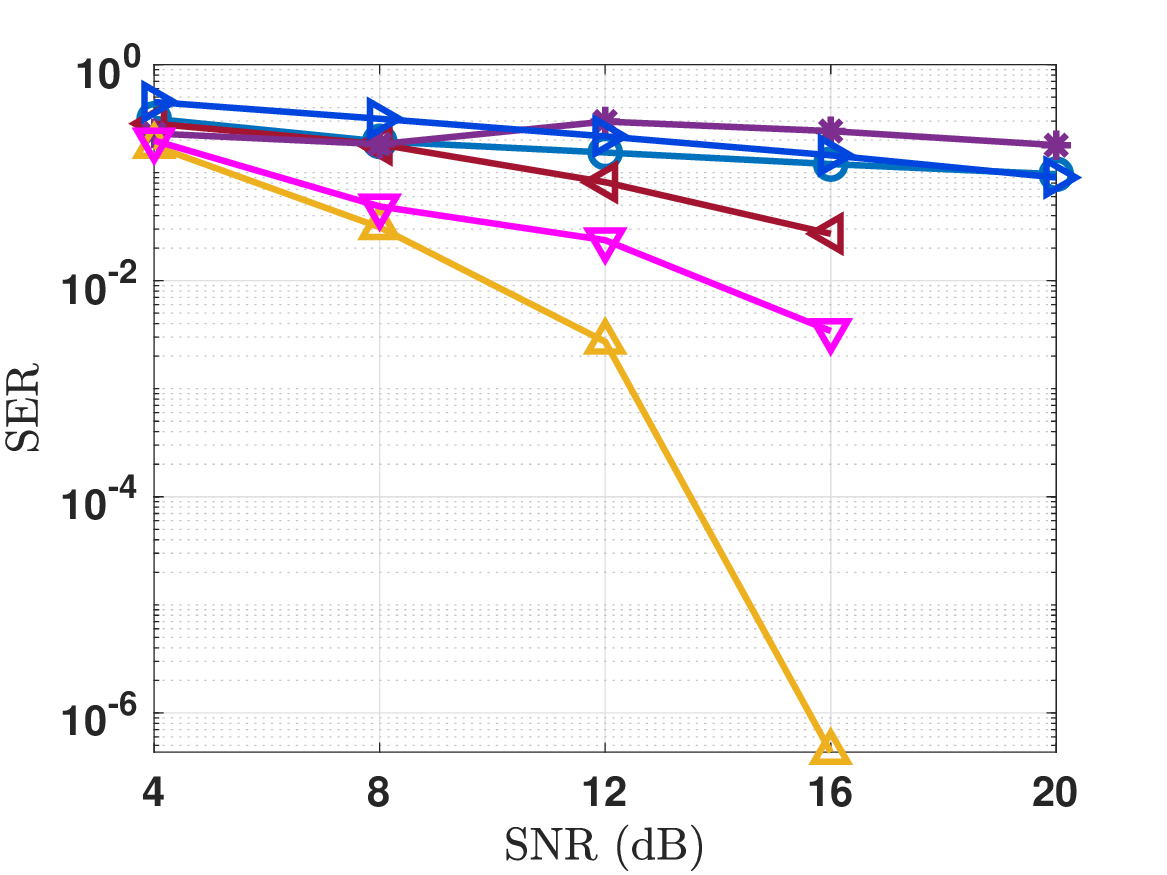}
  \caption{NMSE of IRS-aided channel estimation and SER of different algorithms as functions of SNR}
  \label{fig_ce}
\end{figure}

\begin{table}[h]
\centering
\caption{
Computation time (in seconds) comparison for the IRS-aided MIMO channel estimation with $K_{\mathrm{I}}=6$ and $K_{\mathrm{P}}=3$ (underlined text for \underline{the best} and boldface for \textbf{the second best})
}
\scriptsize
\begin{tabular}{llllll}
\hline
\multicolumn{1}{l|}{SNR}           & \multicolumn{1}{l|}{4dB}                     & \multicolumn{1}{l|}{8dB}                    & \multicolumn{1}{l|}{12dB}                   & \multicolumn{1}{l|}{16dB}                   & 20dB                   \\ \hline \hline
\multicolumn{1}{l|}{PC-KroSBL $\beta = 0.8$} & \multicolumn{1}{l|}{{\ul \textit{32.735}}} & \multicolumn{1}{l|}{15.172}                & \multicolumn{1}{l|}{10.338}               & \multicolumn{1}{l|}{8.559}                & 8.139                \\ \hline
\multicolumn{1}{l|}{PC-KroSBL $\beta = 0.5$} & \multicolumn{1}{l|}{\textbf{36.281}}       & \multicolumn{1}{l|}{\textbf{14.832}}       & \multicolumn{1}{l|}{\textbf{9.391}}       & \multicolumn{1}{l|}{\textbf{7.704}}       & \textbf{7.223}       \\ \hline
\multicolumn{1}{l|}{AM-KroSBL}     & \multicolumn{1}{l|}{54.798}                & \multicolumn{1}{l|}{25.414}                & \multicolumn{1}{l|}{12.103}               & \multicolumn{1}{l|}{10.431}               & 11.312               \\ \hline
\multicolumn{1}{l|}{SVD-KroSBL}    & \multicolumn{1}{l|}{41.954}                & \multicolumn{1}{l|}{{\ul \textit{14.426}}} & \multicolumn{1}{l|}{{\ul \textit{7.472}}} & \multicolumn{1}{l|}{{\ul \textit{6.485}}} & {\ul \textit{6.035}} \\ \hline
\multicolumn{1}{l|}{PC-SBL}        & \multicolumn{1}{l|}{91.534}                & \multicolumn{1}{l|}{86.007}                & \multicolumn{1}{l|}{64.974}               & \multicolumn{1}{l|}{50.231}               & 47.629               \\ \hline
\multicolumn{1}{l|}{cSBL}          & \multicolumn{1}{l|}{84.592}                & \multicolumn{1}{l|}{26.948}                & \multicolumn{1}{l|}{14.638}               & \multicolumn{1}{l|}{11.579}               & 9.250                \\ \hline
\end{tabular}
\label{tab.time_pc}
\vspace{0.2cm}
\end{table}
We focus on the low measurement regime with the under-sampling ratio $\bar{M}/\bar{N} = RK/N^3 \approx 5\%$. Fig. \ref{fig_ce} shows that PC-KroSBL outperforms AM- and SVD-KroSBL that do not account for block sparsity in terms of NMSE and SER. Comparing the performance of PC-KroSBL for different values of $\beta$, we infer that $\beta=0.8$ achieves better accuracy and lower error than $\beta=0.5$ case.
Another interesting observation is that AM-KroSBL has higher NMSE but lower SER than SVD-KroSBL when $\text{SNR} = \{4,8,12,16\}$. This is because SVD-KroSBL fails to identify correct AoAs and AoDs for some channel realizations, resulting in severe channel estimation errors. In such cases, SER is close to one, significantly affecting the overall SER performance of SVD-KroSBL. 
However, the AM-KroSBL does not completely fail even in the low measurement scenario and works better with increasing SNR. Moreover, PC-KroSBL can consistently and accurately estimate the channel and ensure low SER with limited measurements. 

Further, PC-SBL fails in all SNR scenarios because it does not exploit the Kronecker structure and hence, it needs more measurements than KroSBL algorithms for a graceful recovery. Also, while both sparse vectors $\bm g_{\mathrm{L}}$ and $\bm g_{\mathrm{B}}$ exhibit block sparsity, PC-SBL can only exploit the block sparsity of the $\bm g_{\mathrm{L}}\otimes\bm g_{\mathrm{B}}$. In contrast, PC-KroSBL can utilize the block sparsity of both $\bm g_{\mathrm{L}}$ and $\bm g_{\mathrm{B}}$, which is an added advantage of our algorithm. Finally, unlike PC-SBL, the cSBL algorithm involves no approximation, because of which it performs better than PC-SBL in the low measurement regime. However, our KroSBL outperforms cSBL, in which no prior knowledge is incorporated.
 
Table \ref{tab.time_pc} includes the computation time of different schemes. PC-KroSBL has a comparable run time to SVD-KroSBL and is faster than PC-SBL and cSBL despite its iterative M-step (inner loop). This is due to the complexity reduction technique \cite{he2022structure} and the fast convergence of the inner loop. 


\section{Conclusion}

In this paper, we solved the Kronecker-structured linear inversion problem using the SBL framework with the EM procedure. To encourage the Kronecker-supported sparse vector, we designed a prior distribution parameterized by a Kronecker-structured hyperparameter vector and derived two Bayesian algorithms: an iterative AM-based and non-iterative SVD-based algorithm. The AM-based algorithm is theoretically guaranteed to converge to the stationary point of the SBL cost function, while the SVD-based algorithm is faster and has better reconstruction accuracy. We applied our algorithm to the IRS-aided channel estimation for MIMO systems and extended the framework to the block sparsity case. We also analyzed the local minima of the SBL cost function. The analysis of SVD-KroSBL and devising real-time sparse recovery algorithms exploiting the Kronecker structure with reduced complexity can be exciting topics for future work.

 \section*{Acknowledgments}
 We thank Dr. David Wipf, Jianzhe Liang, Changheng Li, and Yanyan Hu for insightful discussions on the theoretical analysis and Sofia-Eirini Kotti for helping us with the coding.

\appendices
\section{Proof of Proposition \ref{thm.am_stationary}}
\label{appe.am}
The proof uses Zangwill's Convergence Theorem A \cite{zangwill1969nonlinear} and the following proposition to arrive at the desired results:

\begin{proposition}\label{thm.coercive_open}
Let $f(\bm x): \mathcal{X}\rightarrow\mathbb{R}$ be a continuous function. If
    \begin{equation}\label{eq.coercive_open}
                f(\bm x)\rightarrow+\infty\ \text{as}\ \mathcal{X} \ni \bm x \rightarrow \mathrm{bd} (\mathcal{X})\ \text{or}\ \|\bm x\| \rightarrow +\infty,
    \end{equation}
    where $\mathrm{bd} (\mathcal{X})$ denotes the boundary of $\mathcal{X}$, then any sublevel set $\mathcal{S}:\{\bm x|f(\bm x)\leq c\}$ is compact.
\end{proposition}

\begin{proof}
This proof is adopted from~\cite{calafiore2014optimization}. A set is compact if it is closed and bounded~\cite{calafiore2014optimization}. We start with the proof of boundedness. Suppose $\mathcal{S}$ is unbounded, then there must exist a sequence $\{\bm x_n\} \subset \mathcal{S}$ with $\|\bm x_n\|\rightarrow +\infty$. Then, \eqref{eq.coercive_open} indicates $f(\bm x_n)\rightarrow +\infty$. Therefore, there exists an integer $n_0>0$ such that $f(\bm x_n) > c$, for $ n > n_0$, which is a contradiction to the assumption $\{\bm x_n\} \subset \mathcal{S}$. So $\mathcal{S}$ is bounded.

We next complete the proof by establishing the closedness of $\mathcal{S}$. Suppose $\mathcal{S}$ is not a closed set. Then, there exists at least one sequence $\{\bm x_n\in \mathcal{S}\}$ converging to $ \bar{\bm x}\notin\mathcal{S}$, which implies either $\bar{\bm x}\in \mathcal{X}$ with $f(\bar{\bm x}) > c$, or $\bar{\bm x} \in \mathrm{bd} (\mathcal{X})$. However, $\bar{\bm x} \in \mathrm{bd} (\mathcal{X})$ further implies that $ f(\bm x_n) \to +\infty$ from \eqref{eq.coercive_open}, which is a contradiction. If  $\bar{\bm x}\in \mathcal{X}$ with $f(\bar{\bm x}) > c$, then there exists a neighbourhood $\mathcal{V}_0$ of $\bar{\bm x}$ such that  $f(\bm x)>c$ for any $\bm x \in \mathcal{V}_0$. The existence of $\mathcal{V}_0$ is guaranteed because $f(\bm x)$ is continuous at $\bar{\bm x}$. However, because $\{\bm x_n\in \mathcal{S}\}$ converges to $ \bar{\bm x}\notin\mathcal{S}$, there exists an integer $n_0$ such that $\bm x_{n} \in \mathcal{V}_0$, for $n > n_0$. Hence, $f(\bm x_n) > c$, leading to a contradiction to the assumption that $\{\bm x_n\} \subset \mathcal{S}$.
Thus, $\mathcal{S}$ contains all its limit points, proving that $\mathcal{S}$ is closed and the proof is complete. 
\end{proof}

Next, we prove the proposition using the above results. We let the sequence $\{\bm\gamma_i\}^{(t)}|_{t=1}^\infty$ be generated by the AM in Algorithm \ref{al.AMKroSBL} with the starting point $\{\bm\gamma_i\}^{(1)} \in \mathcal{C}_+$, where we omit the EM index $r$ for brevity.  We also define the mapping from $\{\bm \gamma_i\}^{(t)}$ to $\{\bm \gamma_i\}^{(t+1)}$ in Algorithm \ref{al.AMKroSBL} as a function $M(\cdot)$, i.e., $M(\{\bm \gamma_i\}^{(t)})=\{\bm \gamma_i\}^{(t+1)}$. We prove the result using the Zangwill's convergence theorem~\cite{zangwill1969nonlinear}. Suppose the following conditions from the Zangwill's convergence theorem hold,
\begin{enumerate}
    \item If the sequence $\{\bm\gamma_i\}^{(t)}|_{t=1}^\infty$ is in a compact set $\mathcal{S} \subset \mathcal{C}_+$. \label{condition1}
    \item If $\{\bm\gamma_i\}^{(t)}$ is not a stationary point of $Q(\{\bm \gamma_i\})$, \label{condition2} 
    \begin{equation}
        Q(\{\bm\gamma_i\}^{(t)}) > Q(\{\bm\gamma_i\}^{(t+1)}).
    \end{equation}
    \item If $\{\bm\gamma_i\}^{(t)}$ is a stationary point of $Q(\{\bm \gamma_i\})$, the AM step terminates or 
    \label{condition3}
    \begin{equation}
        Q(\{\bm\gamma_i\}^{(t)}) \geq Q(\{\bm\gamma_i\}^{(t+1)}).
    \end{equation}
    \item Function $Q(\{\bm \gamma_i\})$ is continuous in $\{\bm \gamma_i\}$ and  $M(\cdot)$ is continuous at \label{condition4} $\{\bm\gamma_i\}^{(t)}$ if $\{\bm\gamma_i\}^{(t)}$ is not a stationary point.
    \end{enumerate}
Then, Zangwill's theorem \cite{zangwill1969nonlinear} guarantees that the AM algorithm stops at a stationary point of $Q(\{\bm \gamma_i\})$. 
Consequently, in the remainder of the proof, we verify Conditions~\ref{condition1}-\ref{condition4}.

We begin with Condition~\ref{condition1}. From Lemma~\ref{thm.am_cost}, we deduce that $\{\bm\gamma_i\}^{(t)}|_{t=1}^\infty\subset \mathcal{S}$ where $\mathcal{S}$ is the sublevel set of $Q(\{\bm\gamma_i\}^{(0)})$. So to check Condition~\ref{condition1}, we prove that $\mathcal{S}$ is compact. To this end, we establish that $\mathcal{S}$ is compact using the following result.
Invoking Proposition~\ref{thm.coercive_open}, it is enough to verify that $Q(\{\bm \gamma_i\})$ satisfies~\eqref{eq.coercive_open} on its domain $\mathcal{C}_+$ to ensure compactness of $\mathcal{S}$. 
For this, we notice that $Q(\{\bm \gamma_i\})$ in~\eqref{eq.qfunc} can be rewritten~as
\begin{align}\label{eq.coerciveT}
    Q(\{\bm \gamma_i\})&=\sum_{j=1}^{N^I}\log [\bm \gamma]_j + [\bm d^{(r)}]_j [\bm \gamma]_j^{-1}\\
    &= \log [\bm \gamma]_{j_*} + [\bm d^{(r)}]_{j_*} [\bm \gamma]_{j_*}^{-1}+\sum_{j\neq j_*}\log [\bm \gamma]_j + [\bm d^{(r)}]_j [\bm \gamma]_j^{-1},
    \end{align}
for any $j_*=1,2,\ldots,N^I$.  Since $\bm d^{(r)}>0$ and the minimum value of function $f(x)=\log x+a/x$ is attained at $x=a$ and $f(x)\geq f(a)=\log a+1$, for any $a> 0$, we get
    \begin{multline}
    Q(\{\bm \gamma_i\})\geq \log [\bm \gamma]_{j_*} + [\bm d^{(r)}]_{j_*} [\bm \gamma]_{j_*}^{-1}+\sum_{j\neq j_*}\left(\log [\bm d^{(r)}]_j + 1\right).
\end{multline} The assumption $\bm d^{(r)}>0$ also implies that the second term in the lower bound $\sum_{j\neq j_*}\left(\log [\bm d]_j + 1\right)$ is finite. Therefore, if  $[\bm \gamma]_{j_*}\to 0$ or $[\bm \gamma]_{j_*}\to +\infty$, the lower bound on $Q$ goes to $+\infty$, and so, $Q\rightarrow +\infty$. 

Clearly, $Q\rightarrow +\infty$ when $ \|\{\bm \gamma_i\}\| \rightarrow +\infty$ because at least one entry $\bm\gamma_{j_*}$ of $\bm \gamma$ satisfies $[\bm \gamma]_{j_*}\to +\infty$. Similarly, when $\mathcal{C}_+ \ni \{\bm \gamma_i\} \rightarrow \mathrm{bd} (\mathcal{C}_+)\ $, there exists an index $i_*$ such that at least one entry of $\bm\gamma_{i_*}$ goes to zero. In that case, we can choose  $j_*$ such that  
\begin{equation}
    [\bm \gamma]_{j_*}=\min_l\;[\bm \gamma_{i_*}]_l\prod_{i\neq i_*}\|\bm \gamma_i\|_\infty.
\end{equation}
Then, $[\bm \gamma]_{j_*}\to 0$ if $\prod_{i\neq i_*}\|\bm \gamma_i\|_\infty$ is finite; otherwise $[\bm \gamma]_{j_*}\to \infty$. In both cases, $Q\rightarrow +\infty$. Hence, \eqref{eq.coercive_open} holds for $Q(\{\bm \gamma_i\})$, and due to Proposition \ref{thm.coercive_open}, $\mathcal{S}$ is compact. Thus, Condition~\ref{condition1} is established.

We next examine Condition~\ref{condition2}. A stationary point  $\{\bm \gamma_i^\mathsf{s}\}\in \mathcal{C}_+$ of $Q$  satisfies 
\begin{equation}
\frac{\partial Q(\{\bm \gamma_i\}|\bm \gamma^{(r)})}{\partial \{\bm \gamma_i\}}\bigg|_{\{\bm \gamma_i^\mathsf{s}\}}
=
\bm 0.
\end{equation}
If $\{\bm \gamma_i\}^{(t)} \in \mathcal{C}_+$ is not a stationary point of $Q(\{\bm \gamma_i\})$, there exists at least one index $i_*\in\{1,2,\ldots,I\}$ such that 
\begin{align}\label{eq.gamma2_ine}
    \bm \gamma_{i_*}^{(t)}&\neq N^{-I+1}\left[(\otimes_{j=1}^{i_*-1}(\bm \gamma_j^{(t)})^{-1})\otimes \bm I_N \otimes (\otimes_{j=i_*+1}^I(\bm \gamma_j^{(t)})^{-1})\right]^\mathsf{T}\bm d^{(r)}\\
  \bm \gamma_i^{(t)} &= N^{-I+1}\big[(\otimes_{j=1}^{i-1}(\bm \gamma_j^{(t)})^{-1})\otimes \bm I_N \otimes (\otimes_{j=n+1}^I(\bm \gamma_j^{(t)})^{-1})\big]^\mathsf{T}\bm d^{(r)},\label{eq.stationary}
\end{align}
for $i=1,2,\ldots,i_*-1$. Therefore, from the AM update \eqref{eq.update}, 
\begin{equation}\label{eq.nochange_am}
    \tilde{\bm \gamma}_i = \bm \gamma_i^{(t)}, \text{ for}\; i=1,2,\ldots,i^{*}-1.
\end{equation}
Substituting this relation in \eqref{eq.update} with $i=i_*$, we obtain
\begin{equation}
    \tilde{\bm \gamma}_{i_*} = N^{-I+1}\big[(\otimes_{i=1}^{i_*-1}(\bm \gamma_i^{(t)})^{-1})\otimes \bm I_N \otimes (\otimes_{i=i_*+1}^I(\bm \gamma_i^{(t)})^{-1})\big]^\mathsf{T}\bm d^{(r)}.
\end{equation}   
The AM update in \eqref{eq.update} also guarantees
\begin{equation}
    \tilde{\bm \gamma}_{i_*}= \arg\min_{\bm\gamma_{i_*}}Q\left(\{\tilde{\bm \gamma}_i\}_{i=1}^{i_*-1},\bm\gamma_{i_*},\{\bm \gamma_i^{(t)}\}_{i=i_*+1}^I\right)\neq  \bm \gamma_{i_*}^{(t)}.
\end{equation}
So using \eqref{eq.nochange_am} and the above relation, we conclude 
\begin{align}
    Q\left(\{\bm \gamma_i\}^{(t)}\right) &=Q\left(\{\tilde{\bm \gamma}_i\}_{i=1}^{i_*-1},\bm \gamma_{i_*}^{(t)},\{\bm \gamma_i^{(t)}\}_{i=i_*+1}^I\right) \\
    &> Q\left(\{\tilde{\bm \gamma}_i\}_{i=1}^{i_*-1},\tilde{\bm \gamma}_{i_*},\{\bm \gamma_i^{(t)}\}_{i=i_*+1}^I\right)  \geq Q(\{\bm \gamma_i\}^{(t+1)}),
\end{align}
using arguments similar to \eqref{eq.nonincrease}. Thus, Condition\ref{condition2} is verified.

 Further, to check Condition~\ref{condition3}, we note that if $\{\bm \gamma_i\}^{(t)} \in \mathcal{C}_+$ is a stationary point, then it satisfies \eqref{eq.stationary} for $i=1,2,\ldots,I$. Then,  
$M\left(\{\bm \gamma_i\}^{(t)}\right) = \{\bm \gamma_i\}^{(t)}$ and $Q\left(\{\bm \gamma_i\}^{(t)}\right) = Q\left(\{\bm \gamma_i\}^{(t+1)}\right)$. So, Condition~\ref{condition3} holds.

Finally, we note that $M(\cdot)$ and $Q$ involves the opertaions $\det(\cdot)$, $\log(\cdot)$, and $(\cdot)^{-1}$ that are continuous on $\mathcal{C}_+$. Hence,  $M(\cdot)$ and $Q(\{\bm \gamma_i\})$ are also continuous, satisfying Condition~\ref{condition4}. Thus, the proof is complete.

\section{Proof of Theorem \ref{thm.stationary}}
\label{appe.em}

Theorem \ref{thm.stationary} is proven using the convergence guarantees for the generalized EM (GEM)~\cite[Theorem 6]{wu1983convergence}. 
GEM is an iterative algorithm similar to EM where the numerically infeasible M-step is replaced with a point-to-set map such that for every iteration $r$,
\begin{equation}
    Q(\{\bm \gamma_i\}^{(r)}|\{\bm \gamma_i\}^{(r)}) \geq Q(\{\bm \gamma_i\}^{(r+1)}|\{\bm \gamma_i\}^{(r)}),
\end{equation}
always holds for $Q$ defined in~\eqref{eq.qfunc}.
Unlike EM, the GEM algorithm does not require achieving the global minimum for the M-step optimization problem. We recall from Lemma~\ref{thm.am_cost} 
AM-KroSBL is a GEM algorithm. Now,     
suppose the following conditions of \cite[Theorem 6]{wu1983convergence} holds,

\begin{enumerate}[leftmargin=*]
    \item the domain of $\mathcal{L}$ is a subset in $\bar{N}$-th dimensional Euclidean space $\mathbb{R}^{\bar{N}}$, \label{gem.c1}
    \item KroSBL cost function $\mathcal{L}$ is continuous in its domain and differentiable in the interior of the domain, \label{gem.c3}
    \item the sublevel set of $\mathcal{L}(\{\bm \gamma_i\}^{(1)})$ is compact for any initilization $\{\bm \gamma_i\}^{(1)}$ with $\mathcal{L}(\{\bm \gamma_i\}^{(1)}) < +\infty$, \label{gem.c2}
    \item the sequence $\{\bm\gamma_i\}^{(r)}|_{r=1}^{\infty}$ generated by AM-KroSBL has the additional property $\nabla Q = \frac{\partial }{\partial \{\bm \gamma_i\}^{(r+1)}}Q(\{\bm \gamma_i\}^{(r+1)}|\bm \gamma^{(r)}) = 0$, \label{gem.c5}
    \item $\nabla Q$ is continuous in both $\{\bm\gamma_i\}^{(r+1)}$ and $\{\bm\gamma_i\}^{(r)}$. \label{gem.c6}
\end{enumerate}
Then, Theorem 6 in \cite{wu1983convergence} ensures that the sequence $\{\bm\gamma_i\}^{(r)}|_{r=1}^{\infty}$ generated by AM-KroSBL converges to the stationary points of $\mathcal{L}$. Here, Condition~\ref{gem.c1} is trivially satisfied because the domain of $\mathcal{L}$, i.e., $\mathcal{C}$, is a subset of the $\bar{N}$-dimensional Euclidean space. Similarly, $\mathcal{L}=\log |\bm \Sigma_{\bm y}| + \bm y^\mathsf{H} \bm \Sigma_{\bm y}^{-1} \bm y$ is a continuous function of $\{\bm \gamma_i\}$ because matrix inversion and determinant are continuous in its entries \cite[Theorems~5.19 and 5.20]{schott2016matrix}. Further, the derivative of $\mathcal{L}$ exists everywhere in $\mathcal{C}$, and thus, Condition~\ref{gem.c3} is also satisfied. So, in the remainder of the proof, we verify if Conditions~\ref{gem.c2}-\ref{gem.c6} to establish the convergence guarantee of AM-KroSBL.

To verify Condition~Conditions~\ref{gem.c2}, we first show the compactness of the sublevel set of $\mathcal{L}$. Compactness is established via the coerciveness of the KroSBL cost function $\mathcal{L}$ since the sublevel sets of coercive functions are compact \cite{calafiore2014optimization}.
By definition~\cite{fedorov2018structured}, function $\mathcal{L}$ is coercive if $ \lim_{\|\{\bm \gamma_i\}\|\to+\infty} \mathcal{L} = +\infty,$
where $\|\{\bm \gamma_i\}\| = \left(\sum_{i=1}^I\|\bm \gamma_i\|_2^2\right)^{1/2}$. 
Further, since $\bm \Sigma_{\bm y} = \sigma^2 \bm I_{\bar{M}} + \bm H \bm \Gamma \bm H^\mathsf{H}$ is positive-definite (PD) when $\sigma^2 > 0$, we have
\begin{equation}
    \mathcal{L} = \log |\bm \Sigma_{\bm y}| + \bm y^\mathsf{H} \bm \Sigma_{\bm y}^{-1} \bm y
    >\log |\bm \Sigma_{\bm y}|
    =\sum_{j=1}^{\bar{M}} \log (\sigma^2 + \lambda_j),
\end{equation}
where $\lambda_j\!\geq\!0$ is the $j$th eigenvalue of $\bm H \bm \Gamma \bm H^\mathsf{H}$. So $\mathcal{L}$ is coercive~if
\begin{equation}\label{eq.sumloglimit}
    \lim_{\|\{\bm \gamma_i\}\|\to+\infty} \sum_{j=1}^{\bar{M}} \log (\sigma^2 + \lambda_j) = +\infty,
\end{equation}
which is true if at least one of the $\lambda_j$'s goes to infinity as $\|\{\bm \gamma_i\}\|$ goes to $+\infty$. Moreover, from the boundedness assumption on the norm of the dictionary columns, we derive
\begin{equation*}\label{eq.trace_lb_infty}
\sum_{j=1}^{\bar{M}} \lambda_j\!=\!\trace(\bm H \bm \Gamma \bm H^\mathsf{H})=\sum_{i=1}^{N^I}[\bm\gamma]_i\|[\bm H]_i\|_2^2\geq \epsilon^2\sum_{i=1}^{N^I}[\bm\gamma]_i\geq \epsilon^2\|\bm\gamma\|_{\infty}. 
\end{equation*}
So, $\lim_{\|\{\bm \gamma_i\}\|\to+\infty} \sum_{j=1}^{\bar{M}} \lambda_j=+\infty$, 
which means at least one eigenvalue goes to $+\infty$, proving \eqref{eq.sumloglimit}. 
Thus, KroSBL cost function $\mathcal{L}$ is a coercive function on $\mathcal{C}$, and its sublevel set of the starting point $\{\bm \gamma_i\}^{(1)}$ is compact. As a consequence, Condition~\ref{gem.c2} is true.

We next verify Condition~\ref{gem.c5}, which requires the point $\{\bm \gamma_i\}^{(r+1)}$ to be a stationary point of $Q(\{\bm \gamma_i\}^{(r+1)}|\bm \gamma^{(r)})$. According to Proposition \ref{thm.am_stationary}, it holds if $\bm d^{(r)} > 0$. So, we next show that $\bm d^{(r)} = \diag(\bm \Sigma_{\bm x}+\bm \mu_{\bm x}\bm \mu_{\bm x}^\mathsf{H}) > 0$. Since $\diag(\bm \mu_{\bm x}\bm \mu_{\bm x}^\mathsf{H})\geq  0$, it suffices to show that $\diag(\bm \Sigma_{\bm x})>0$. Also, since PD matrices have positive diagonal entries, from \eqref{eq.post_meva}, it is enough to verify that $\bm \Sigma_{\bm x}^{-1} = \sigma^{-2}\bm H^\mathsf{H}\bm H+\diag(\bm \gamma^{(r)})^{-1}$ is PD. So, $\bm d^{(r)} > 0$ if $\{\bm \gamma_i\}^{(r)} >0$ because $\sigma^{-2}\bm H^\mathsf{H}\bm H$ is PSD \cite{million2007hadamard}. From this observation, we prove the condition $\bm d^{(r)} > 0$ using induction. Since $\{\bm \gamma_i\}^{(1)} \in\mathcal{C}_+$, we have $\bm d^{(1)}>0$. Next, we assume that $\bm d^{(r)}>0$, for some $r>0$. From Proposition~\ref{thm.am_stationary}, $\{\bm \gamma_i\}^{(r+1)}$ generated by the AM algorithm is bounded away from $\mathrm{bd}(\mathcal{C}_+)$ when $\bm d^{(r)} > 0$. As a result, we conclude that $\{\bm \gamma_i\}^{(r+1)}>0$, which in turn implies that $\bm d^{(r+1)}>0$. Hence, Condition~\ref{gem.c5} holds in our case.

Finally, we show Condition~\ref{gem.c6} by first computing the gradient of $Q$ with respect to $\bm \gamma_i$ as
\begin{multline}
    \frac{\partial }{\partial \bm \gamma_i}Q(\{\bm \gamma_i\}|\bm \gamma^{(r)})=-N^{I-1}\bm\gamma_i^{-1}+\diag(\bm\gamma_i)^{-2} \\\times\left[\left(\otimes_{l=1}^{i-1}(\bm \gamma_l)^{-1}\right)\otimes \bm I_N \otimes \left(\otimes_{l=i+1}^I(\bm \gamma_l^{(r,t)})^{-1}\right)\right]^\mathsf{T}\bm d^{(r)}.
\end{multline}
Here, operators $(\cdot)^{-1}$ and $(\cdot)^{-2}$ are continuous in $(0,+\infty)$. Thus, the gradient is continuous in $\{\bm \gamma_i\}$. Finally, in $\nabla Q$, only $\bm d^{(r)}$ depends on $\{\bm \gamma_i\}^{(r)}$ as in~\eqref{eq.post_meva}. Since the matrix inversion is continuous in its entries \cite[Theorem 5.20]{schott2016matrix}, $\bm \Sigma_{\bm x}$ and $\bm \mu_{\bm x}$ are continuous in $\{\bm \gamma_i\}^{(r)}$. Therefore, $\nabla Q$ is continuous in $\{\bm \gamma_i\}^{(r)}$. Thus, Condition~\ref{gem.c6} is established, and the proof is complete.

\section{Proof of Theorem~\ref{thm.local_minima_sparse}}\label{appe.sparse_local}

To prove the sparsity of local minima, we start with a few supporting lemmas.

\begin{lemma}\label{lmm.ml_separable}
$\log|\bm \Sigma_{\bm y}|$ is concave with respect to $\{\bm \gamma_i\}$ in the noiseless case, i.e., $\sigma^2=0$.
\end{lemma}

\begin{proof}
    When $\sigma^2 = 0$, $\bm \Sigma_{\bm y} = \bm H \bm \Gamma \bm H^\mathsf{H} = \otimes_{i=1}^I \bm H_i \bm \Gamma_i \bm H_i^\mathsf{H}$. We have
    \begin{equation}
        \log|\bm \Sigma_{\bm y}| = \log|\otimes_{i=1}^I \bm H_i \bm \Gamma_i \bm H_i^\mathsf{H}|=\sum_{i=1}^I\left({\prod_{j\neq i}M_j}\right)\log|\bm H_i \bm \Gamma_i \bm H_i^\mathsf{H}|.
    \end{equation}
Since $\bm H_i \bm \Gamma_i \bm H_i^\mathsf{H}$ is a PSD matrix and affine in $\bm \gamma_i$, and function $\log|\cdot|$ is a concave function in the space of PSD matrices, $\log|\bm H_i \bm \Gamma_i \bm H_i^\mathsf{H}|$ is concave in $\bm \gamma_i$. Thus, $\log|\bm \Sigma_{\bm y}|$ is concave because the sum of concave functions is also concave.
\end{proof}

\begin{lemma}\label{lmm.constant}
If $\{\bm \gamma_i\}$ satisfies $\bm b = \bm A (\otimes_{i=1}^I\bm \gamma_i)$, where 
    \begin{equation}\label{eq:ab_defn}
        \bm b = \bm y - \sigma^2\bm u;\ \ \ \bm A = \bm H\diag(\bm H^\mathsf{H}\bm u),
    \end{equation}
    and $\bm u$ is any fixed vector such that $\bm y^\mathsf{H}\bm u = C$, then $\bm y^\mathsf{H}\bm \Sigma_{\bm y}^{-1}\bm y$ is a constant $C$ for any value of $\sigma^2\geq 0$.
\end{lemma}

\begin{proof}

We combine the relation $\bm b = \bm A (\otimes_{i=1}^I\bm \gamma_i)$ and \eqref{eq:ab_defn} to obtain
\begin{align}
\bm y &= \bm A (\otimes_{i=1}^I\bm \gamma_i) + \sigma^2\bm u =\bm H\diag(\bm H^\mathsf{H}\bm u) (\otimes_{i=1}^I\bm \gamma_i) + \sigma^2\bm u \\
&= \bm H \bm \Gamma \bm H^\mathsf{H} \bm u + \sigma^2\bm u = \bm \Sigma_{\bm y}\bm u,
\end{align}    
where $\bm \Gamma=\diag(\otimes_{i=1}^I\bm \gamma_i)$. Then, $\bm y^\mathsf{H} \bm \Sigma_{\bm y}^{-1} \bm y= \bm y^\mathsf{H}\bm u=C$, for any value $\sigma^2\geq 0$.

\end{proof}



    

\begin{lemma}\label{lmm.kron_equations_separable}
    Consider the set of linear equations, with $\bm t_1,\bm t_2\neq \bm 0$,
    \begin{equation}\label{eq.linear_kron}
        \left(\bm \Phi_1 \otimes \bm \Phi_2 \right)\left(\bm w_1\otimes \bm w_2 \right) = \bm t_1 \otimes \bm t_2.
    \end{equation}
    Seeking $\bm w_1$ and $\bm w_2$ that satisfy~\eqref{eq.linear_kron} is equivalent to solving 
    \begin{equation}\label{eq.linear_kron_soln}
        \bm \Phi_1 \bm w_1 = \alpha\bm t_1;\ \bm \Phi_2 \bm w_2 = \alpha^{-1}\bm t_2,
    \end{equation}
    where $\alpha$ is any non-zero scalar.
\end{lemma}

\begin{proof}
    We rewrite~\eqref{eq.linear_kron} as $(\bm \Phi_1\bm w_1)\otimes(\bm \Phi_2\bm w_2)=\bm t_1\otimes\bm t_2$. Now, we can arrange this equation as $\left(\bm \Phi_2\bm w_2 \right)\left(\bm \Phi_1\bm w_1 \right)^{\mathsf{H}}=\bm t_2 \bm t_1^\mathsf{H}$, which is a rank-one matrix. Here, $\bm t_2 \bm t_1^\mathsf{H}$ has at least one nonzero column since $\bm t_1,\bm t_2\neq \bm 0$, and every column of $\bm t_2 \bm t_1^\mathsf{H}$ is a scaled version $\bm t_2$. Therefore, the solution to the system of equations is given by \eqref{eq.linear_kron_soln},
leading to the desired conclusion.
\end{proof}

With all the mentioned lemmas, we present the proof of Theorem~\ref{thm.local_minima_sparse}. First, we pose another optimization problem:
\begin{equation}\label{problem}
    \underset{\substack{\{\bm \gamma_i\}\geq 0}}{\min}\log|\bm \Sigma_{\bm y}| \;\text{s.t.} \;
    \bm A \left(\otimes_{i=1}^I\bm \gamma_i\right) = \bm b,
\end{equation}
where $\bm A = \bm H\diag(\bm H^\mathsf{H}\bm u)$ and $\bm b = \bm y - \sigma^2\bm u$. As Lemma \ref{lmm.constant}, the constraint $\bm A (\otimes_{i=1}^I\bm \gamma_i) = \bm b$ holds the second term of $\mathcal{L}$ constant, and we minimize the first term of $\mathcal{L}$, which is concave, over a bounded convex polytope.
Then, any local minimum of~\eqref{eq.ml_problem} denoted by $\{\bm \gamma_i^*\}$, must also be a local minimum of~\eqref{problem} with
\begin{equation}
    \bm y^\mathsf{H}\left(\sigma^2 \bm I_{\bar{M}} + \bm H \diag(\otimes_{i=1}^I \bm \gamma_i^*) \bm H^\mathsf{H}\right)^{-1}\bm y = C^* = \bm y^\mathsf{H}\bm u^*,
\end{equation}
as long as there exists a vector $\bm u^*$ satisfying $C^*\!=\!\bm y^\mathsf{H}\bm u^*$ to construct $\bm A$ and $\bm b$. A candidate for $\bm u^*$ in the noiseless setting~is
\begin{equation}
    \bm u^* = C^*\frac{ \bm y}{\|\bm y\|_2^2} =C^*\frac{ \otimes_{i=1}^I\bm H_i \bm x_i}{\|\bm y\|_2^2},
\end{equation}
where we also use \eqref{eq.problem_basic}, \eqref{eq.separable_dict} and the assumption that $\bm x =\otimes_i^I\bm x_i$. Thus, we obtain 
\begin{equation*}
 \bm A =  \frac{C^*}{\|\bm y\|_2^2}  \bm H\diag(\bm H^\mathsf{H}\otimes_{i=1}^I\bm H_i \bm x_i)=\frac{C^*}{\|\bm y\|_2^2}  \otimes_{i=1}^I\bm H_i\diag(\bm H_i^\mathsf{H}\bm H_i \bm x_i).
\end{equation*}
Similarly, we have
\begin{equation}
    \bm b = \bm y = \otimes_i^I\bm H_i\bm x_i.
\end{equation}
Thus, using Lemma \ref{lmm.kron_equations_separable} with the scaling factor $\alpha=1$,  the constraint $\bm A (\otimes_{i=1}^I\bm \gamma_i)=\bm b$ can be written as 
\begin{equation}
    \bm A_i \bm \gamma_i = \bm b_i,\forall i=1,2,\ldots,I,
\end{equation}
where $\bm A_i=\left(C^*/\|\bm y\|_2^2\right)^{1/I}  \bm H_i\diag(\bm H_i^\mathsf{H}\bm H_i \bm x_i)$, and $\bm b =\bm H_i\bm x_i$. 
Now, the problem \eqref{problem} can be transformed into $I$ separated problems as
\begin{equation}\label{problem_2}
    \underset{\substack{\bm\gamma_i}}{\min}\ \log|\bm H_i \bm \Gamma_i \bm H_i^\mathsf{H}| \hspace{0.3cm} \text{s.t.} \
    \bm A_i \bm \gamma_i = \bm y_i, \bm \gamma_i \geq 0,
\end{equation}
for $i=1,2,\ldots,I$. Any local minimum of $\mathcal{L}$, e.g., $\{\bm \gamma_i\}^*$ is also a local minimum of \eqref{problem_2}. Furthermore, all local minima of \eqref{problem_2} are achieved at extreme points, which are also basic feasible solutions with at most ${M_i}$ non-zeros for each $\bm \gamma_i$ \cite{luenberger1984linear}. Thus, $\{\bm \gamma_i^*\}$ is sparse when noise is absent.

\section{Proof of Theorem~\ref{thm.no_local}}\label{appe.upper_bound}

The proof is based on the following lemma:
\begin{lemma}\label{prop.kron_urp}
    If $\bm H = \otimes_{i=1}^I \bm H_i$ satisfies URP, then $\bm H_i$ for all $i$ also satisfy URP.
\end{lemma}

\begin{proof}
    Suppose we choose any {$M_i$} columns of $\bm H_i$ indexed by $\mathcal{M}_i$, for $i=1,2,\ldots,I$. Then, the Kronecker product of the corresponding submatrices, given by $\otimes_{i=1}^I [\bm H_i]_{\mathcal{M}_i}\in\mathbb{C}^{\bar{M}\times\bar{M}}$, is a submatrix of $\bm H=\otimes_{i=1}^I \bm H_i$.  Since $\bm H$ satisfies URP, any subset of $\bar{M}$ columns of $\bm H$ are linearly independent, we get  
    \begin{equation}\label{eq.URP_rank}
        \bar{M}=\rank\big([\bm H]_{\otimes_{i=1}^I \mathcal{M}_i}\big) = \prod_{i=1}^I \rank\left([\bm H_i]_{\mathcal{M}_i}\right)\leq {\prod_{i=1}^I M_i}={\bar{M}}.
    \end{equation} The last step follows because $\rank\big([\bm H_i]_{\mathcal{M}_i}\big) \leq {M_i}$. So, \eqref{eq.URP_rank} holds if and only if $\rank\big([\bm H_i]_{\mathcal{M}_i}\big) = {M_i}$ for $i=1,2,\ldots,I$. Hence, any subset of ${M_i}$ columns of $\bm H_i$ is linearly independent, and $\bm H_i$ satisfies URP, for $i=1,2,\ldots,I$.
\end{proof}

Lemma \ref{prop.kron_urp} ensures that $\bm H_i$ satisfies the URP for all $i$. Then, we next show that for any index set $\mathcal{M}_i$ such that $|\mathcal{M}_i|={M_i}$, if we restrict the nonzero values of $\bm \gamma_i$ to the set $\mathcal{M}_i$, there can only be one minimum for $\mathcal{L}(\{[\bm \gamma_i]_{\mathcal{M}_i})$. For this, we note that \begin{align}
    \mathcal{L}(\{[\bm \gamma_i]_{\mathcal{M}_i}) & = \log \left|\otimes_{i=1}^I [\bm H_i]_{\mathcal{M}_i} [\bm \Gamma_i]_{\mathcal{M}_i} [\bm H_i]_{\mathcal{M}_i}^{\mathsf{H}} \right| \notag\\
    &\hspace{0.5cm}+ \bm y^\mathsf{H} \left( \otimes_{i=1}^I [\bm H_i]_{\mathcal{M}_i} [\bm \Gamma_i]_{\mathcal{M}_i} [\bm H_i]_{\mathcal{M}_i}^{\mathsf{H}}\right)^{-1} \bm y\\
   & = \sum_{i=1}^I \left({\prod_{j\neq i}M_j}\right) \log \left|[\bm H_i]_{\mathcal{M}_i} [\bm \Gamma_i]_{\mathcal{M}_i} [\bm H_i]_{\mathcal{M}_i}^{\mathsf{H}} \right|\notag\\
    &\hspace{0.5cm}+ \bm y^\mathsf{H} \otimes_{i=1}^I\left( \left([\bm H_i]_{\mathcal{M}_i}^{\mathsf{H}}\right)^{-1}  [\bm \Gamma_i]_{\mathcal{M}_i}^{-1}[\bm H_i]_{\mathcal{M}_i}^{-1} \right) \bm y.
    \label{eq.Costfnt_simpl}
\end{align}
Here, the second term can be simplified using the assumption, $\bm y = \left(\otimes_{i=1}^I[\bm H_i]_{\mathcal{M}_i}\right)\left(\otimes_{i=1}^I[\bm x_i]_{\mathcal{M}_i} \right) = \otimes_{i=1}^I([\bm H_i]_{\mathcal{M}_i}[\bm x_i]_{\mathcal{M}_i})$, as follows: 
\begin{multline}
  \bm y^\mathsf{H} \otimes_{i=1}^I\left(  [\bm H_i]_{\mathcal{M}_i} [\bm \Gamma_i]_{\mathcal{M}_i} [\bm H_i]_{\mathcal{M}_i}^{\mathsf{H}}\right)^{-1} \bm y\\= \otimes_{i=1}^I[\bm x_i]_{\mathcal{M}_i}^\mathsf{H}  [\bm \Gamma_i]_{\mathcal{M}_i}^{-1} [\bm x_i]_{\mathcal{M}_i} = \prod_{i=1}^I[\bm x_i]_{\mathcal{M}_i}^\mathsf{H}  [\bm \Gamma_i]_{\mathcal{M}_i}^{-1} [\bm x_i]_{\mathcal{M}_i}.
    \end{multline}
Therefore, from \eqref{eq.Costfnt_simpl}, we arrive at
\begin{multline}
  \mathcal{L}(\{[\bm \gamma_i]_{\mathcal{M}_i})= \sum_{i=1}^I {\prod_{j\neq i}M_j} \log \left|[\bm H_i]_{\mathcal{M}_i} [\bm H_i]_{\mathcal{M}_i}^{\mathsf{H}} \right|\\+  \sum_{i=1}^I {\prod_{j\neq i}M_j} \log \left|[\bm \Gamma_i]_{\mathcal{M}_i}  \right|+ \prod_{i=1}^I[\bm x_i]_{\mathcal{M}_i}^\mathsf{H}  [\bm \Gamma_i]_{\mathcal{M}_i}^{-1} [\bm x_i]_{\mathcal{M}_i}.
    \end{multline}
    Setting the derivative of the above function with respect to $\bm\gamma_i$ to zero gives
    \begin{equation}
         {\prod_{j\neq i}M_j} [\bm \gamma_i]_{\mathcal{M}_i}  =  [\bm x_i]_{\mathcal{M}_i}^2  \prod_{j\neq i}[\bm x_j]_{\mathcal{M}_j}^\mathsf{H}  [\bm \Gamma_j]_{\mathcal{M}_j}^{-1} [\bm x_j]_{\mathcal{M}_j},
    \end{equation}
    for $i=1,2,\ldots,I$. However, since $\{\bm \gamma_i\}\in\mathcal{C}$, we obtain the unique minimum of $ \mathcal{L}(\{[\bm \gamma_i]_{\mathcal{M}_i})$ at
    \begin{equation}
        [\bm \gamma_i]_{\mathcal{M}_i}=\begin{cases}
            \frac{[\bm x_i]_{\mathcal{M}_i}^2}{\|[\bm x_i]_{\mathcal{M}_i}^2\|}, & \text{if } i<I\\
            [\bm x_I]_{\mathcal{M}_I}^2  \prod_{j=1}^{I-1}\|[\bm x_j]_{\mathcal{M}_j}^2\|  & \text{if } i=I.
        \end{cases}
    \end{equation}

Therefore, every set of $\{\mathcal{M}_I\}_{i=1}^I$ corresponds to one unique local minimum, and counting for all possible index set combinations, we get $\mathcal{N}\leq \prod_{i=1}^I\binom{N}{{M_i}}$. Further, all index set combinations containing a degenerate solution share the same minimum~\cite{wipf2004sparse}. Accounting for such repetitions, we derive~\eqref{eq.upper_local_minima}. Thus, the proof is complete.

\bibliographystyle{IEEEtran}
\bibliography{bibfile}


 




\vfill

\end{document}